\documentclass[a4paper,12pt]{article}
\usepackage[14pt]{extsizes} 
\usepackage[mathscr]{eucal}
\usepackage{amsmath,amsfonts,amssymb,amsthm}
\usepackage{times}
\usepackage{hyperref}

\newcommand{\maxim}[1]{\textsc{#1}}

\numberwithin{equation}{section} 

\newtheorem{thm}{Theorem}[section]
\newtheorem{rem}[thm]{Remark}

\newtheorem{theorem}[thm]{\textbf{Theorem}}
 \newtheorem{prop}[thm]{Proposition}
 
 \newtheorem{definition}[thm]{Definition}
 \newtheorem{example}[thm]{Example}

\renewcommand{\tilde}{\widetilde}
\renewcommand{\hat}{\widehat}

\newcommand{\bref}[1]{\textbf{\ref{#1}}}

\newcommand{\p}[1]{|#1|}
\newcommand{\gh}[1]{\mathrm{gh}(#1)}

\newcommand{\dx}{\mathrm{d}_X}
\newcommand{\dy}{\mathrm{d}_Y}

\renewcommand{\d}{\partial}
\renewcommand{\dh}{\mathrm{d_h}}

\newcommand{\cF}{\mathcal{F}}

\newcommand{\tensor}{\otimes}

\renewcommand{\geq}{\,{\geqslant}\,}

\newcommand{\inner}[2]{\langle #1{,}\,#2\rangle}
\newcommand{\binner}[2]{%
  {\langle}\kern-4.15pt{\langle}#1{,}\,#2{\rangle}\kern-4.15pt{\rangle}}
\newcommand{\commut}[2]{[#1{,}\,#2]}

\newcommand{\qcommut}[2]{[#1{,}\,#2]_*}

\newcommand{\half}{\mathchoice{%
    \ffrac{1}{2}}{\frac{1}{2}}{\frac{1}{2}}{\frac{1}{2}}}

\newcommand{\ffrac}[2]{\raisebox{.5pt}%
  {\footnotesize$\displaystyle\frac{#1}{#2}$}\kern1pt}

\newcommand{\dl}[1]{\mathchoice{\ffrac{\d}{\d #1}}{\frac{\d}{\d #1}}{\ffrac{\d}{\d #1}}{\ffrac{\d}{\d #1}}}

\newcommand{\st}[2]{{\overset{#1}{#2}}}

\newcommand{\ddl}[2]{\ffrac{\d #1}{\d #2}}

\newcommand{\Liealg}{\mathfrak} 
\newcommand{\algg}{\Liealg{g}}

\newcommand{\algp}{\Liealg{p}}
\newcommand{\algh}{\Liealg{h}}

\newcommand{\algA}{\mathcal{A}}

\newcommand{\cC}{\mathcal{C}}

\newcommand{\fC}{\mathbb{C}}

\newcommand{\fR}{\mathbb{R}}

\newcommand{\fZ}{\mathbb{Z}}

 \def\cE{\mathcal{E}}

 \def\cI{\mathcal{I}}
 
\def\cK{\mathcal{K}}
 \def\cL{\mathcal{L}}
 \def\cX{\mathcal{X}}

\newcommand\blfootnote[1]{%
  \begingroup
  \renewcommand\thefootnote{}\footnote{#1}%
  \addtocounter{footnote}{-1}%
  \endgroup
}

\usepackage[left=3cm,right=2cm,
    top=3cm,bottom=3cm,bindingoffset=0cm]{geometry}

\newcommand{\dg}{\mathrm{d}_\algg}
\newcommand{\comm}[1]{}
\usepackage{autobreak}
\usepackage{authblk}
\usepackage{tikz-cd}

\newcommand{\lorcon}{w}

\author{Ivan Dneprov~}

\author{~Maxim Grigoriev${}^{\dagger}$}

\affil{\textsl{Service de Physique de l'Univers, Champs et Gravitation, \protect\\ Universit\'e de Mons, 20 place du Parc, 7000 Mons, 
Belgium \vspace{5pt}}}

\date{}

\begin{document}

\title{Background fields in the presymplectic BV-AKSZ approach}
\maketitle

\begin{abstract}
The Batalin-Vilkovisky formulation of a general local gauge theory can be encoded in the structure of a so-called presymplectic gauge PDE -- an almost-$Q$ bundle over the spacetime exterior algebra, equipped with a compatible presymplectic structure. In the case of a trivial bundle and an invertible presymplectic structure, this reduces to the well-known AKSZ sigma model construction. We develop an extension of the presympletic BV-AKSZ approach to describe local gauge theories with background fields. It turns out that such theories correspond to presymplectic gauge PDEs whose base spaces are again gauge PDEs describing background fields. As such, the geometric structure is that of a bundle over a bundle over a given spacetime. Gauge PDEs over backgrounds arise naturally when studying linearisation, coupling (gauge) fields to background geometry, gauging global symmetries, etc. Less obvious examples involve parametrised systems, Fedosov equations, and the so-called homogeneous (presymplectic) gauge PDEs. The latter are the gauge-invariant generalisations of the familiar homogeneous PDEs and they provide a very concise description of gauge fields on homogeneous spaces such as higher spin gauge fields on Minkowski, (A)dS, and conformal spaces. Finally, we briefly discuss how the higher-form symmetries and their gauging fit into the framework using the simplest example of the Maxwell field.

\end{abstract}

\vfill

\blfootnote{Supported by the ULYSSE Incentive
Grant for Mobility in Scientific Research [MISU] F.6003.24, F.R.S.-FNRS, Belgium.}

\blfootnote{${}^{\dagger}$ Also at Lebedev Physical Institute and Institute for Theoretical and Mathematical Physics, Lomonosov MSU, Moscow, Russia}

\newpage
\tableofcontents

\section{Introduction}

In applications one often encounters gauge field theories where some of the fields are not dynamical but still it is useful to keep them at the equal footing with the others and not to set them to a fixed configuration from the outset. The typical example of a background field is the background metric which serves as a background for a genuine dynamical (gauge) field or a collection of fields. It can be also useful to artificially decompose the fields into the background and the perturbation as is done in the background field method~\cite{Abbott1982}.  Models with background fields typically arise when one studies coupling of (gauge) fields to background geometry, gauging global symmetries, or expanding a given theory around generic configurations. 

A systematic and first-principle method to study general gauge theories is the Batalin-Vilkovisky (BV) formalism~\cite{Batalin:1981jr,Batalin:1983wj}. In the case of local gauge theories considered at the level of equations of motion,  a modern and rather flexible extension of BV that allows to maintain manifest locality is the BV-AKSZ approach based on representing a given gauge theory in terms of a so-called gauge PDE~\cite{Grigoriev:2019ojp,Barnich:2010sw,Barnich:2004cr}. Gauge PDE (gPDE) is a $\fZ$-graded $Q$-bundle~\cite{Kotov:2007nr} over the shifted tangent bundle over the spacetime. When the underlying bundle is globally trivial and finite-dimensional, the corresponding gauge PDE reduces to the non-Lagrangian version~\cite{Barnich:2006hbb} of the conventional AKSZ construction~\cite{Alexandrov:1995kv} for topological theories. Another extreme case is when fibres do not have coordinates of nonvanisging degree and the notion reduces to the usual PDE defined intrinsically in terms of the infinitely-prolonged equation manifold equipped with the Cartan distribution~\cite{Vinogradov:1977,Krasil?shchik-Lychagin-Vinogradov}. 
If the system is diffeomorphism-invariant and negative degree variables are not present, the gPDE equations of motion have the form a free differential algebra~\cite{Sullivan:1977fk} so that gPDE approach can also be seen as a generalisation of the so-called unfolded formulation~\cite{Vasiliev:1980as,Lopatin:1987hz,Vasiliev:2005zu}. A gauge PDE encodes the full-scale BV formulation of the underlying gauge theory. In particular the usual jet-bundle BV formulation, see e.g.~\cite{Barnich:1995ap,Barnich:2000zw}, of the underlying theory is contained in the super-jet bundle of the corresponding gPDE.  

As far as Lagrangian theories are concerned, the additional structure appears to be a compatible presymplectic structure, giving a notion of the presymplectic gauge PDE. An interesting feature is that in the presymplectic case one can relax $Q^2=0$ condition by requiring $Q^2$ to lie in the kernel distribution of the presymplectic structure~\cite{Grigoriev:2022zlq,Dneprov:2024cvt}, see also~\cite{Alkalaev:2013hta,Grigoriev:2016wmk,Grigoriev:2020xec}. This allows one to encode nontopological gauge theories in terms of finite-dimensional presymplectic gPDEs, leading to a rather concise formulations of various models.\footnote{An alternative modification of the AKSZ construction necessary to cover non-topological theories is based on replacing the base space exterior algebra  with a more general differential graded commutative algebra~\cite{Bonechi:2009kx,costello2011renormalization,Bonechi:2022aji}. Another possibility discussed in the literature has to do with introducing extra structures in the target space~\cite{Arvanitakis:2025nyy}.}

In this work we systematically extend the (presymplectic) gauge PDE approach to the case of theories with background fields. It turns out that such theories  are naturally described in this approach by (presymplectic) gauge PDEs whose base spaces are themselves gauge PDEs  for background fields, giving the notion of a (presymplectic) gPDE over background. While at the level of equations of motion this extension is rather clear and to some extent is known in the literature~\cite{Shaynkman:2004vu,Ponomarev:2010ji,Barnich:2010sw}, the Lagrangian version is less straightforward.

In the Lagrangian case, the crucial observation is that compatibility between $Q$-structure and the presymplectic structure should hold only modulo the differential ideal generated by forms pulled back from the background gPDE while the presymplectic master equation remains intact. These conditions ensure that the BV-AKSZ master-action satisfies master equation if background fields are subject to their own equations of motion encoded in the background gPDE. Moreover, we demonstrate that gauge symmetries associated to background field are determined by the total $Q$-structure and can be explicitly expressed in terms of $Q$. Note that our approach shares some similarity to the equivariant version of the BV formalism put forward in~\cite{Bonechi:2019dqk,Mnev:2012qd,Bonechi:2022aji,Cattaneo:2024ows}.

There are a few ways in which gauge PDEs over backgrounds emerge naturally. The most standard way is by gauging global symmetries. In the gPDE approach, this amounts to simply adding ghost variables associated with symmetries and extending the $Q$-structure by the corresponding BRST differential. Each new ghost variable gives rise to a 1-form (or higher-form, in case of lower-degree symmetries) gauge field that is interpreted as a background field.  In the Lagrangian setting, when considering symmetries that do not affect the presymplectic structure, this procedure applies in a straightforward way. Nevertheless, the resulting presymplectic gPDE over the background determined by the ghosts may acquire a classical anomaly.

In the more intricate case of space-time symmetries, interesting gPDEs over background that can be regarded as the result of gauging such symmetries are the so-called homogeneous gPDEs over background. If the background solution is fixed, these are generalisations of the conventional homogeneous PDEs which are known to be fully characterised by specifying a homogeneous spacetime $G/H$ and a suitable $\algg=Lie(G)$-space (a typical fiber of the resulting PDE), see e.g.~\cite{Eastwood,cap2009parabolic}. Remarkably, a straightforward generalisation, where the typical fiber is in addition equipped with a $\algg$-invariant $Q$-structure, is known to provide a concise description of a very general class of gauge fields on homogeneous spaces. For instance, generic higher-spin gauge fields in flat, AdS, and conformal spaces admit a concise description in terms of homogeneous gPDEs~\cite{Barnich:2004cr,Barnich:2006pc,Alkalaev:2008gi,Bekaert:2009fg,Alkalaev:2009vm,Bekaert:2013zya,Chekmenev:2015kzf}. The respective gPDE over background is obtained by extending the base space $T[1]X,X=G/H$ into a $Q$-bundle whose typical fiber
is $\algg[1]$ and whose solutions are flat Cartan connections. It turns out that the presymplectic version is even more flexible and requires $F$ to be a $\algg$-space modulo the kernel distribution of the presymplectic structure only. Homogeneous gPDEs can be used as a starting point for studying the coupling of (gauge) fields to a curved (in the sense of Cartan) geometry. 

The paper is organized as follows. Section~\bref{sec:prelim} contains definitions, basic properties, and simplest examples of (presymplectic) gPDEs. There we also recall the notion of weak gPDEs~\cite{Grigoriev:2024ncm} which are nonlagrangian counterparts of presymplectic gPDEs where $Q^2$ is allowed to belong to an involutive distribution, giving a possibility to describe non-topological systems in terms of finite-dimensional bundles. In Section~\bref{sec:bgpde} we introduce and study non-Lagrangian systems with background fields. In particular, we introduce a notion of gPDE over background and show how such objects naturally arise from gauging global symmetries and linearizing about generic configurations. Homogeneous gPDEs over background are also introduced there and the example of Fronsdal theory of massless higher-spin fields is given. Section~\bref{sec:presymp-bgpde} is devoted to Lagrangian theories with background fields in the presymplectic BV-AKSZ approach. After defining presymplectc gPDEs over background we demonstrate that their background symmetries are controlled by the total $Q$-structure which determine gauge symmetries of the total AKSZ-like action involving background fields. Just like in the non-Lagrangian setup, background system naturally arise from linearization and gauging global symmetries. However,
the gauging is only straightforward in the case where the initial global symmetries do not affect the presymplectic structure. In the case where they do, we limit our discussion to homogeneous presymplectic gPDEs which are to be understood as a result of gauging spacetime symmetries.



\section{Preliminaries}\label{sec:prelim}

\subsection{(Weak) $Q$-manifolds}
In this work we mostly use the language of graded geometry. More precisely, the basic objects are graded (super)manifolds. The $\fZ$-degree determines the 
cohomological degree of the underlying BRST complexes and is referred to as the \textit{ghost degree}. In general, if the system under consideration involves physical anticommuting variables there is an additional fermionic degree denoted by $\epsilon(\cdot)$
 so that the total Grassmann parity of a homogeneous function $f$ is given by $\p{f}=\gh{f}+\epsilon(f)$
 and the supercommutativity relation reads as $fg=(-1)^{\p{f}\p{g}}gf$.
To simplify the exposition, we assume that the fermionic degree is trivial. It can always be reinstated if fermions are present. Unless otherwise specified, by vector fields, forms, bundles and so on we mean the corresponding objects in the graded geometry setup.

Before passing to geometrical structure underlying field theories with background fields let us recall several important building blocks.
\begin{definition}
An almost $Q$-manifold is a pair $(M,Q)$, where $M$ is a graded manifold and $Q$ is a ghost degree $1$ vector field on $M$. $(M,Q)$ is called a $Q$-manifold if, in addition, $Q$ is homological, i.e. $Q^2=\half\commut{Q}{Q}=0$.
\end{definition}

$Q$-manifolds can be regarded as gauge systems in the spacetime dimension $0$. Indeed, the zero locus of $Q$ can be interpreted as solutions
while $Q$-exact vector fields define gauge transformations and their higher analogs. This is precisely how $Q$-manifolds emerge in the context of gauge systems.

A typical example of $Q$-manifolds, which we often encounter in this work, originates from a Lie algebra $\algg$ action $\rho$ on a smooth manifold $F$. Let $e_\alpha$ denote a basis in $\algg$ and $V_\alpha=\rho(e_\alpha)$ the fundamental vector field on $F$. Then there is a natural $Q$-structure on $\algg[1]\times F$, where $\algg[1]$ denotes a shifted $\algg$, i.e. an algebra of functions on $\algg[1]$ is simply the exterior algebra of $\algg^*$ or, in other words, functions in degree $1$ coordinates $c^\alpha$ corresponding to basis elements $e_\alpha$. Then the $Q$ structure is given by:
\begin{equation}
\label{CEdiff}
\dg=-\half c^\alpha c^\beta U_{\alpha\beta}^\gamma \dl{c^\gamma}+c^\alpha V_\alpha\,, \qquad \commut{e_\alpha}{e_\beta}=U_{\alpha\beta}^\gamma  e_\gamma\,.
\end{equation}
Functions on $\algg[1]\times F$ can be identified with Chevalley-Eilenberg (CE) cochains with coefficients on $\cC^\infty(F)$ while $\dg$ is the CE differential. 

Another standard example of $Q$ manifold which we use extensively is the shifted tangent bundle  $T[1]X$  of the real smooth manifold $X$. Functions on $T[1]X$ can be identified with differential forms on $X$ and under this identification the de Rham differential on $X$ is sent to a homological vector field $\dx$ defined on $T[1]X$.

In the case of field theory with local degrees of freedom, the underlying $Q$-manifolds typically become infinite-dimensional. However, there is an alternative approach in which a field theory can be described by
the finite-dimensional almost $Q$-manifold where condition $Q^2=0$ is relaxed in a controllable way. This leads to the concept of a weak $Q$ manifold~\cite{Grigoriev:2024ncm}:
\begin{definition}
A weak $Q$-manifold is an almost $Q$-manifold equipped with an involutive distribution $\cK$ such that
\begin{equation}
    L_Q\cK \subseteq \cK\,, \qquad Q^2\in \cK\,.
\end{equation}
\end{definition}
Here and below by distribution on $M$ we mean a finitely generated submodule of vector fields on $M$, seen as a module over $\cC^\infty(M)$. If the submodule happens to be of a constant rank, i.e. it is freely generated locally, we call it regular. If $\algA\subset \cC^\infty(M)$ denotes the subalgebra of functions annihilated by $\cK$ one finds that $Q^2f=0$ for all $f\in\algA$. In particular, if $\cK$ is a vertical distribution of a fiber bundle $M\to N$ then $N$ is naturally a $Q$-manifold. This can be regarded as an implicit way to define $Q$-manifolds. 

In the context of Lagrangian systems one in the first place is interested in symplectic $Q$-manifolds as they underlie Batalin-Vilkovisky formalism and its variations such as AKSZ construction or the Hamitonian Batalin-Fradkin-Vilkovisky approach. The correspoding counterpart of the concept of a weak $Q$-manifold is given by~\cite{Grigoriev:2022zlq} (see also~\cite{Dneprov:2024cvt})
\begin{definition}
A presymplectic $Q$-manifold $(M,Q,\omega)$ is an almost $Q$-manifold equipped with a presymplectic structure $\omega$, such that
\begin{equation}
    L_Q\omega=0\,, \qquad i_Qi_Q\omega=0\,.
\end{equation}
If $\omega$ is symplectic then $Q^2=0$ and hence the notion reduces to that of a symplectic $Q$-manifold. 
\end{definition}
Recall, that presymplectic structure is a closed 2-form which we also assume to have a definite ghost degree.
It is easy to see that taking $\cK$ to be the kernel distribution for $\omega$, one finds that $(M,Q,\cK)$ is a weak $Q$-manifold. Moreover, if $\cK$ originates from the fibration $M\to N$ then $N$ is naturally a symplectic $Q$-manifold, see \cite{Grigoriev:2020xec, Dneprov:2024cvt,Grigoriev:2024ncm} for more details.

\subsection{(Weak) gauge PDEs}

Let us now turn to field theory. If we are interested in local field theory, i.e. equations of motion are PDEs and gauge symmetries are determined by differential operators, then a natural generalization of the description based on $Q$-manifolds can be achieved in terms of \textit{$Q$-bundles}~\cite{Kotov:2007nr} over $T[1]X$, where $X$ is a space-time manifold. This leads to the concept of so-called \textit{gauge PDEs} introduced in~\cite{Grigoriev:2019ojp} (see also \cite{Barnich:2010sw,Barnich:2004cr} for the earlier but less geometrical version of this notion):
\begin{definition}
    A graded fibre bundle $E \xrightarrow{\pi} T[1]X$ equipped with a homological vector field $Q: deg(Q) =1, Q^2=0$ s.t. $Q \circ \pi^* = \pi^* \circ \dx$ is called a \textit{gauge PDE} and is denoted $(E,Q,T[1]X)$.  Sections of $E$ are interpreted as fields and solutions of a gauge PDE $(E,Q,T[1]X)$ are sections satisfying:
    \begin{equation} \label{gPDE-sol}
        \sigma^* \circ Q = \dx \circ \sigma^*.
    \end{equation}
    Gauge transformations are generated by the vector fields of the form $\commut{Q}{Y}$, where the gauge parameter $Y$, $\gh{Y}=-1$ is a projectable  vector field on $E$. The corresponding transformation of $\sigma^*$ can be written as 
    \begin{align}\label{gaugetransf}
        \delta_Y\sigma^{*}={\sigma}^{*}\circ \commut{Q}{Y}-\commut{\dx}{y}\circ \sigma^*\,,
    \end{align}
    where $y\equiv \pi_* Y$ is the projection of $Y$ to $T[1]X$. Gauge for gauge symmetries are defined in an analogous way.
\end{definition}
In addition, it is usually assumed that  $(E,Q,T[1]X)$ is locally equivalent, in the sense of a weak equivalence as introduced in~\cite{Grigoriev:2019ojp}, to a nonnegatively graded $Q$-bundle. Moreover, if $E$ is infinite-dimensional one should also require that it is locally equivalent to a gauge PDE arising from a local BV system. This assumption is needed to exclude nonlocal systems and is a straightforward extension of the analogous condition known in the geometric theory of PDEs. Note that the action of a gauge symmetry on a given section can also be realised in terms of a vertical gauge parameter, at least locally. Nevertheless, it is often convenient to allow for gauge parameters that are not necessarily vertical.

\begin{rem}\label{sol-subman} 
    A solution $\sigma:T[1]X \to E$ can be seen as a submanifold in $E$ such that its projection to $T[1]X$ is a diffeomorphism and $\sigma\subset E$ is a $Q$-submanifold of $E$, i.e. $Q$ is tangent to $\sigma\subset E$. Indeed, if $f\in \cC^\infty(E)$ vanishes on $\sigma$ then  $\sigma^*(f)=0$ and vice versa. Becuase $\sigma$ is a solution $\dx \sigma^*(f)=\sigma^* (Qf)=0$ so that $Qf$ vanishes on $\sigma$ and hence $Q$ is tangent to $\sigma$.
\end{rem}
It is also instructive to see this in terms of local coordinates. Let $x^a,\theta^a,b^i$ be local coordinates on $E$ such that $x^a,\theta^a$ are the adapted coordinates on $T[1]X$ pulled back to $E$. Section $\sigma$ seen as a submanifold in $E$, is the zero locus of the following constraints:
\begin{equation}
    \pi^* \circ \sigma^* (b^i)
 - b^i.
\end{equation}
Applying $Q$ and using $\dx \circ \sigma^*=\sigma^*\circ Q$ one finds that the result is proportional to the constraints. The above remark also applies to generic $Q$-bundles. Namely, $Q$-sections of a  $Q$-bundle  are sections that are $Q$-submanifolds.

\begin{example}
Consider a gPDE of the form $E=(F,q)\times (T[1]X,\dx)$ seen as a bundle over $(T[1]X,\dx)$. This  gives the nonlagrangian version of the AKSZ sigma model. Indeed, if $(F,q)$ is a symplectic $Q$-manifold then this data define an AKSZ sigma model with source $(T[1]X,\dx)$ and target $(M,q,\omega)$. Without symplectic structure this should be naturally regarded as an AKSZ model at the level of equations of motion~\cite{Barnich:2006hbb}.
\end{example}
\begin{example}
Let $E_0\to X$ be a PDE defined in the intrinsic terms. Namely, $E_0$ is the fibre bundle equipped with the involutive horizontal distribution $C\subset TE$ known as Cartan distribution. The algebra of horizontal forms on $E_0$ can be identified with the algebra of functions on the bundle $E\to T[1]X$ which is $E_0$ pulled back to $T[1]X$ by the canonical projection $T[1]X\to X$. Under this identification the horizontal differential $\dh$ defines a $Q$-structure on $E$. It is easy to check that in this way one arrives at gPDE whose solutions coincides with the solutions of the starting point PDE. In other words, usual PDEs can be naturally considered as gPDEs. See~\cite{Grigoriev:2019ojp} for further details.
\end{example}
Note that any gPDE defines a PDE which singles out $Q$-sections among all its sections, while the information about gauge invariance is simply forgotten. However, this correspondence is not functorial in the sense that two equivalent gPDEs generally define inequivalent PDEs. For instance, the system obtained by adding an algebraically pure-gauge variables (known as Stueckelberg fields) have more solutions but the space of gauge-inequivalent solutions remains unchanged.

In fact a rather general local gauge theory can be equivalently represented as a gPDE, at least locally. The corresponding gPDE can be constructed starting from the jet-bundle BV formulation of the theory, as was originally shown in~\cite{Barnich:2010sw}, see also \cite{Grigoriev:2019ojp} for a more geometrical discussion in present terms. Note that now we are discussing gauge systems at the level of equations of motion. In particular, it is important to stress that gauge PDEs can be employed to describe off-shell field theories, i.e. systems that can be equivalently represented in such a way that equations of motion are absent. In such a representation the corresponding jet-bundle BV description is not going to contain negative ghost-degree coordinates so that the BV-BRST differential does not involve Koszul-Tate piece and hence does not impose any equations. Of course, any off-shell system can be equivalently represented as the one with nontrivial equations of motion by introducing auxiliary fields.

Typically, a fiber bundle $E$ underlying the gPDE formulation of a nontopological system is inifinite-dimensional and in the case of nonlinear systems is usually defined only implicitly. It turns out, that 
by  employing a fiber-bundle version of weak $Q$-manifolds one can describe nontopological theories in terms of finite-dimensional geometrical objects. More precisely, these are so-called weak gPDEs \cite{Grigoriev:2024ncm}:
\begin{definition} \label{def:weak-gPDE}
    1. A weak gPDE $(E,Q,\cK,T[1]X)$ is a $\mathbb{Z}$-graded bundle $\pi: E \rightarrow T[1]X$ equipped with a vector field $Q$ such that $\gh{Q} = 1$, $Q \circ \pi^* = \pi^* \circ \dx$ and a vertical involutive distribution $\cK$ satisfying:
    i) $L_Q \cK \subseteq \cK$
    ii) $Q^2 \in \cK$ 
    iii) $\cK$ is quasi-regular 
    
    2. A solution of $(E,Q,\cK,T[1]X)$ is a section $\sigma:T[1]X \rightarrow E$ such that the degree 1 vector field along $\sigma$ defined as $R_\sigma := \dx \circ \sigma^* - \sigma^* \circ Q$, satisfies
    $R_\sigma \in \sigma^*\cK$, where $\sigma^*\cK$ denotes the pullback of the distribution $\cK$.
    
    3. Infinitesimal gauge transformations are generated by the vector fields of the form $\commut{Q}{Y}$ where the gauge parameter $Y$, $\gh{Y}=-1$ is a projectable  vector field on $E$ satisfying $L_Y \cK \in \cK$. The corresponding transformation of $\sigma^*$ can be written as 
    \begin{align}\label{gaugetransf2}
        \delta_Y\sigma^{*}={\sigma}^{*}\circ \commut{Q}{Y}-\commut{\dx}{y}\circ \sigma^*\,,
    \end{align}
    where $y\equiv \pi_* Y$ is the projection of $Y$ to $T[1]X$.  In a similar way one defines gauge for gauge symmetries.

    4. Two solutions differing by an algebraic gauge equivalence generated by $\delta_{alg} \sigma^* = \sigma^* \circ k$ where $k\in \cK$, $\gh{k} = 0$, are considered equivalent.
\end{definition}
The condition of quasi-regularity in iii) means that the prolongation of $\cK$ to superjets is regular, details can be found in \cite{Grigoriev:2024ncm}. Note that the algebraic equivalence described in item 4. can be nontrivial only if $\cK$ has a nontrivial subdistribution of degree $0$. In applications it is often possible to reformulate the system in such a way that degree $0$ component of $\cK$ vanishes and to avoid taking the quotient when defining solutions. In particular, the weak gPDEs considered in this work are of this type.

Let us recall how projectable vector fields can be defined in terms of the algebras of functions on the corresponding manifolds. A vector field $V$ on the total space of a bundle $E \xrightarrow[]{p} M$ is called projectable if there exists a vector field $v$ on $M$, such that $\forall f \in C^\infty(M)$ one has $V (p^* f) = p^* (vf)$.

\begin{definition}
    A projectable vector field $V$, $\gh{V} = 0$ is called a symmetry of a weak gPDE $(E,Q,\cK)$ if
    \begin{equation}
        [V,Q] \in \cK, \qquad [V,\cK] \subseteq \cK\,.
    \end{equation}
    The corresponding infinitesimal transformation of section $\sigma:T[1]X\to E$ is given by
    $\delta \sigma^*=\sigma^* \circ V - v\circ \sigma^*$, where $v=\pi_*V$ denotes the projection of $V$ to $T[1]X$. As usual, two symmetries are considered equivalent if they differ by a gauge one. 
\end{definition}
It is an immediate consequence of the definition that symmetries maps solution to solutions. Let us stress that in this work we only discuss infinitesimal (gauge) symmetries.  In the special case when $\cK$ is empty this reduces to the usual definition of gPDE symmetries~\cite{Grigoriev:2019ojp,Grigoriev:2023kkk}. \comm{\footnote{\maxim{Note that in the case of gPDEs any symmetry can be made vertical by adding a suitable gauge symmetry. This is generally not the case when we consider weak gPDEs, because vector fields generating gauge transformations have to preserve $\cK$. Generally, the vector field $\ddl{}{\theta^a}$ does not preserve $\cK$.}}} Note that in contrast to gPDEs, this definition does not generally cover higher derivative symmetries of weak gPDEs. In order to see them, one should allow for a suitable analog of generalized vector fields.

\subsection{Presymplectic gPDEs}\label{sec:presymplectic}
We are mostly interested in Lagrangian theories. The first-principle systematic framework to handle general gauge field theories is that of the BV formalism. If one insists on the manifest locality of the approach a suitable version of BV is based on the following:

\begin{definition}\label{def:local-BV}
    A local BV  system  with the underlying fiber bundle $\cE\to X$,  $\dim(X)=n$,
    is determined by the following data:\\
(i) a degree-1, evolutionary vector field $s$ defined on $J^\infty(\cE)$ and satisfying $s^2=0$\\
(ii) an $(n,2)$-form $\st{n}{\omega}\in \bigwedge^{(n,2)}(J^\infty(\cE))$ of ghost degree $-1$, which is a pullback to $J^\infty(\cE)$ of a closed $n+2$ form $\omega^\cE$ on $\cE$, 
such that
\footnote{Here and in what follows $L_W \equiv i_W d +(-1)^{\p{W}} di_W$ denotes the Lie derivative along the vector field $W$.}
\begin{equation}
\label{eq:s-inv}
L_s\st{n}{\omega}+\dh (\ldots)=0\,, 
\end{equation}
where $\dh$ is the horizontal part of the de Rham differential on $J^\infty(\cE)$. In addition, $\omega^\cE$ is required not to have zero-vectors.  
\end{definition}
Further details on the jet-bundle BV formalism can be found in e.g.~\cite{Barnich:2000zw}, see also~\cite{Grigoriev:2022zlq} for the exposition in the present language.

Local BV systems can be encoded into more concise geometrical objects that usually can be assumed finite-dimensional. These objects can be considered as a field theoretical extension of the concept of presymplectic $Q$-manifolds or as a Lagrangian analogs of weak gPDEs. More precisely:
\begin{definition} \label{Def:w_presymp_gPDE}
A presymplectic gauge PDE $(E,Q,\omega)$ is a $\mathbb{Z}$ graded fiber bundle $E \xrightarrow[]{\pi} T[1]X$ equipped with a 2-form $\omega$ of 
degree $n-1$, a 0-form $\cL$ of degree $n$ 
and a vector field $Q$ of degree 1 satisfying $Q\circ \pi^* = \pi^* \circ \dx$ and 
\begin{equation}
\label{wgPDE}
    d\omega = 0, \qquad i_Q\omega + d\mathcal{L} \in \mathcal{I}, \qquad \frac{1}{2} i_Q i_Q\omega + Q\mathcal{L} = 0\,,
\end{equation}
where $\cI \subset \bigwedge^\bullet(E)$ is the ideal generated by the forms $\pi^*\alpha, \alpha \in \bigwedge^{k > 0}(T[1]X)$.
\end{definition}
It is known that these data define the full-scale BV formulation, provided the presymplectic structure obeys certain regularity assumptions. In more details, the BV field-antifield space arises as follows: one considers jets of supersectons $J^\infty{E}$ which is equipped with the prolongation of $Q$. Moreover,
$\omega$ induces a density-valued presymplectic structure of ghost degree $-1$ on $J^\infty{E}$ restricted to $X\subset T[1]X$. If this happen to be a regular, one can take the symplectic quotient. It follows that the prolongation of $Q$ projects to this quotient and together with the symplectic structure defines a structure of a local BV system. Moreover, the corresponding BV action has the AKSZ-like form and reads as:
\begin{equation}
\label{BV-AKSZ}
S_{BV}(\hat\sigma)=\int_{T[1]X}\hat\sigma^*(\chi)(\dx)+\hat\sigma^*(\cL)
\end{equation}
where $\chi$ is the presymplectic potential $\omega=d\chi$ and $\hat\sigma$ denotes a supersection. This actions does not depend on fields associated to the kernel directions of the induced presymplectic structure on $J^\infty(E)$ and hence is well-defined on the quotient. The detailed explicit construction of the local BV system encoded in a presymplectic gPDE can be found in~\cite{Dneprov:2024cvt}, see also~\cite{Grigoriev:2020xec,Dneprov:2022jyn,Grigoriev:2022zlq}.  We also refer to the formalism based on representing gauge systems in terms of presymplectic gPDEs as to presymplectic BV-AKSZ approach.

The restriction of $S_{BV}$ to sections, i.e. to  degree-preserving maps, gives the classical action as a functional of section $\sigma:T[1]X\to E$. It turns out that the gauge symmetries of this action can be explicitly represented in terms of $Q$. Namely, let $Y$ be a projectable vector field on $E$ and $y=\pi_*Y$ its projection to $T[1]X$.  Consider again the transformation \eqref{gaugetransf2}, i.e. 
\begin{align}\label{gaugetransf3}
        \delta_Y\sigma^{*}={\sigma}^{*}\circ \commut{Q}{Y}-\commut{\dx}{y}\circ \sigma^*\,.
\end{align}
One can check that it defines a symmetry provided $L_Y\omega\in \cI$, i.e. the gauge parameter preserves the symplectic structure. This was shown in~\cite{Alkalaev:2013hta} in a less general setup and we prove a more general statement in Section~\bref{sec:capturing}.

\subsubsection{Brackets of Hamiltonians on presymplectic bundles}

The defining relations of a presymplectic gPDE could appear somewhat confusing as $Q$ is Hamiltonian modulo $\cI$ only and, moreover, instead of the conventional condition $Q^2=0$ one has a somewhat unusual equation $i_Qi_Q\omega+2Q\cL=0$. 

It turns out that one can reformulate the latter condition as a conventional master-equation by introducing a suitable bracket on the space of generalised Hamiltonians.  To see this let us first make some general observations about brackets on bundles equipped with presymplectic structures. 
\begin{definition} \label{hamiltonian-def}
    Let $E \xrightarrow{\pi} \mathcal{X}$ be a (graded) bundle with a presymplectic structure $\omega$, $d\omega=0$. A pair $(y, J_y)$, where $y \in \mathfrak{X}(\mathcal{X})$, $J_y \in C^\infty(E)$ is called a projectable Hamiltonian pair if there exists a projectable vector field $Y\in \mathfrak{X}({E})$ such that 
    \begin{equation}\label{hamilt-projection}
        \pi_*Y = y\,,
    \end{equation}
    \begin{equation} \label{hamiltonian-vf}
        i_Y\omega +  dJ_Y \in \cI_{\mathcal{X}}\,,
    \end{equation}
    where $\cI_{\mathcal{X}}$ is the differential ideal generated by 1-forms of the form $\pi^*\alpha$, $\alpha \in \bigwedge^1(\cX)$.
\end{definition}
In the case where $\cX$ is just a point, the above definition reproduces the usual notion of Hamiltonian functions on a presymplectic manifold. It is clear that   if $(y, J_y)$ is a projectable Hamiltonian pair then so is $(y, J_y + \pi^*f)$, $f\in \cC^\infty(\cX)$.

\begin{prop}\label{bracket-for-pairs-def}
    The space of projectable Hamiltonian pairs admits a natural Lie bracket defined by 
    \begin{multline} \label{bracket-for-pairs}
        \{(x,J_X), (y, J_Y)\} = 
        ([x,y], ((-1)^{|X|}i_Xi_Y\omega +\\ (-1)^{|X|}L_X J_Y - (-1)^{(|X|+1)|Y| }L_Y J_X))\,,
    \end{multline}
    where $X, Y$ are any vector fields (of fixed ghost degree) on $E$ obeying \eqref{hamilt-projection}, \eqref{hamiltonian-vf}. In particular,
    the bracket does not depend on the choice of $X,Y$. 
\end{prop}
The  proof is done by a straightforward computation. Related brackets were considered in various contexts~\cite{Deligne:1999qp,Fiorenza_2014,delgado2018lagrangianfieldtheoriesindproapproach,schiavina2025homotopieslagrangianfieldtheory}. However, they generally require extra structures.

Let us now turn to a very special but important case of $\cX\equiv T[1]X$ so that $\cX$ is endowed with a canonical vector field $\dx$. Then the structure of presymplectic gPDE is simply given by a presymplectic structure and a projective Hamiltonian pair $(\dx,\cL)$ satisfying the conventional-looking master equation:
\begin{equation}
        \{(d_X,\cL), (d_X,\cL) \} = 0\,.
\end{equation}
Indeed, it easy to check that this just a rewriting of $\half i_Qi_Q\omega+Q\cL=0$.
This observation is quite useful because it shows that the space of projective Hamiltonian pairs on a presymplectic gPDE is naturally a homological complex despite that $Q^2\neq 0$ in general. In particular, this explains why the deformation theory of presymplectic gPDEs developed in~\cite{Frias:2024vgw} works pretty much the same way as the usual deformation theory in terms of differential graded Lie algebras.

\subsubsection{Presymplectic BV-AKSZ for gravity}\label{example:presymp-grav}
Here we recall the typical example of Palatini-Cartan-Weyl gravity in the presymplectic gPDE approach~\cite{Alkalaev:2013hta,Grigoriev:2020xec}. The underlying $Q$-bundle $E\to T[1]X$ is $\algg[1]\times T[1]X$ with $Q=\dx+\dg$, where $\algg$ is the Poincar\'e algebra and $\dg$ its CE differential seen as a homological vector field on $\algg[1]$. If $\xi^a, \rho_{ab}$ are coordinates on $\algg[1]$ associated to translocations and rotation, respectively, one has 
\begin{equation}
        \dg\xi^a = -\rho^{a}{ }_b \xi^b \,, \qquad \dg \rho^{ab} = -\rho^{a}{ }_b \rho^{bc}\,.
\end{equation}
The presymplectic structure on $E$ is taken to be 
\begin{equation}\label{grav-presymp}
    \omega = -\epsilon_{abcd}\xi^a d\xi^b d\rho^{cd} =d{\xi}^{(2)}_{ab} d \rho^{ab}\,,
\end{equation}
where for simplicity we restricted ourselves to the case of $\dim X=4$. Generalisation to any $\dim X \geq 3$
and to the case of a nonvanishing cosmological constant is straightforward and can be found in~\cite{Alkalaev:2013hta,Grigoriev:2020xec}. It is easy to check that all the axioms are satisfied, giving
\begin{equation}
\cL = \half \epsilon_{abcd}\xi^a\xi^b \rho^{ce}\rho_e{}^d = {\xi}^{(2)}_{cd}\rho^{ce}\rho_e{}^d \,.
\end{equation}
The resulting action is precisely the Palatini-Cartan-Weyl action written in terms of the frame field $e^a=\sigma^*(\xi^a)$ and Lorentz connection $\lorcon^{ab}=\sigma^*(\rho^{ab})$:
\begin{equation}
  S[e,\lorcon]=\int_{T[1]X}e^{(2)}_{ab}(\dx \lorcon^{ab} + \lorcon^{a}{}_c\lorcon^{cb}) \,,
\end{equation}
while its BV extension is given by \eqref{BV-AKSZ}. Strictly speaking, to obtain the usual BV formulation one needs to eliminate fields parameterizing the kernel of the presymplectic structure by e.g. setting them to some fixed values.

In the above formulas and below we use the following useful convention: 
\begin{equation}
    {\xi}^{(n-k)}_{a_{n-k+1},...a_n} = \frac{1}{(n-k)!} \epsilon_{a_1,\cdots a_{n-k} a_{n-k+1}\cdots a_n} \xi^{a_1}\cdots \xi^{a_{n-k}}\,,
\end{equation}
where $n = dim X$.
Analogous notations are used for other anticommuting vectors, for instance ${\theta}^{(2)}_{ab} = \half \epsilon_{pkab}\theta^p \theta^k$ if $n=4$.

\subsection{Alternative interpretation of presymplectic gPDEs} \label{sec:altern_interpr}

It is instructive to discuss an alternative interpretation of presymplectic gPDEs in the case where the underlying almost $Q$-bundle $(E,Q,T[1]X)$ is actually a $Q$-bundle. In this situation one can consider two apparently different gauge systems: (i) the Lagrangian gauge theory encoded in the presymplectic gPDE (ii) the non-Lagrangian gauge system determined by $(E,Q,T[1]X)$
seen as a gauge PDE. It may happen that these two are equivalent if one considers (i) as a system at the level of equations of motion. In this case we say that the presymplectic structure on gPDE is complete. In general, this is not the case. The instructive example is given by the presymplectic BV-AKSZ formulation of gravity recalled in Section~\bref{example:presymp-grav}. Seen as a gPDE this formulation simply describes flat $\algg$-valued connections modulo the natural gauge-equivalence and is of course a topological system. However, seen as a presymplectic gPDE with presymplectic structure~\eqref{grav-presymp}, it describes Einstein gravity if $\dim{X}\geq 3$. Note that for $\dim{X}>3$ it is not topological. 

Another perspective on the same relation is as follows: given a gauge PDE equipped with a compatible presymplectic stucture, i.e. $L_Q\omega\in \cI$ one finds that there exits $\cL$ such that $i_Q\omega+d\cL\in\cI$. In this case $Q^2=0$ implies that $\frac{1}{2} i_Q i_Q\omega + Q\mathcal{L} = 0$ so that we are dealing with a special case of a presymplectic gPDE. Indeed, 
\begin{multline}
    i_{Q^2}\omega = 0 = L_Q i_Q\omega - i_Q L_Q \omega = L_Q i_Q\omega + i_Q d i_Q \omega  =\\ 2L_Qi_Q \omega + di_Qi_Q\omega  = d(i_Qi_Q\omega + 2Q\cL) + \cI 
    \end{multline}
which implies $i_Qi_Q\omega + 2Q\cL = \pi^*f$ for some $f\in C^\infty(T[1]X)$. But since the expression is of degree $n+1$ it can only be zero. This suggests that a gPDE equipped with a compatible presymplectic structure can be regarded as a partially Lagrangian gauge system. Indeed, as we discussed above, seen as a presymplectic gPDE it generally has more inequivalent solutions then the underlying gPDE or in other words not all of the equations encoded in this gPDE are variational. A recent discussion of partially Lagrangian systems within a different framework along with some examples can be found in~\cite{Lyakhovich:2024yai}.

Analogous considerations apply to the case of weak gPDEs, see the discussion in~\cite{Grigoriev:2024ncm}.

\section{Gauge PDEs over background} \label{sec:bgpde}

Although in this work we are mostly interested in Lagrangian systems, it is instructive first to consider systems at the level of equations of motion. 

\subsection{gPDEs over background and symmetries} \label{subsec-bgrd-gPDE}

As defined above, (weak) gPDEs are bundles over $T[1]X$. This is a natural minimal construction because it does not require any additional geometrical structures to be defined on $X$. Indeed, differential forms and de Rham differential are defined for any smooth manifold. In the gPDE setup all extra structures on $X$ (if any) are encoded in $Q$.

If, on the contrary, we are interested in describing a gauge theory on the spacetime manifold equipped with certain geometrical structures, one may try to replace $T[1]X$ with a more general object.  The important observation, is that various geometrical structures can be understood as gauge theories, usually off-shell ones. For instance, at the infinitesimal level, Riemannian geometry can be seen as a gauge theory where the gauge field is the metric and the gauge transformations are infinitesimal diffeomorphisms. But a gauge theory describing background geometry can be again reformulated as a gauge PDE. This leads to the concept of gauge PDE over background, which we introduce and study in this section.

Let $\pi_B:(E,Q) \rightarrow (B,\gamma)$ be a $Q$-bundle where the base $B$ is itself a nontrivial gauge PDE $(B,\gamma,T[1]X)$ over $T[1]X$. A simple observation is that any solution $S$ of $(B,\gamma,T[1]X)$ immediately gives a new gPDE $(E|_{S},Q,T[1]X)$ over $T[1]X$ which is a pullback of $E$ to $S$ seen as a submanifold of $B$. It is easy to check that this is a gauge PDE. Indeed, $Q$ is tangent to $(\pi_B)^{-1}S\subset E$ (see Remark~\ref{sol-subman} ) and $(S,\gamma|_S)$ is isomorphic to $(T[1]X,\dx)$, with the isomorphism map being the projection $\pi_{T[1]X}:B \rightarrow T[1]X$ restricted to $S$. We call such objects \textit{gauge PDEs over background}. 

Note that a gauge PDE over background can be also seen as a usual gPDE over $T[1]X$ with additional structure.  Namely, its total space $E$ is itself a bundle over $T[1]X$ and the projection $E\to B$
is compatible with the bundle structure. More precisely, if $\pi:E\to T[1]X$ and
$\pi_{T[1]X}:B\to T[1]X$ are the corresponding projections then $\pi = \pi_{T[1]X}\circ \pi_B$. In other words, $\pi_B$ is a morphism of bundles over $T[1]X$ which induces the identity map of the base.

The structure of gPDE over background naturally defines a certain subalgebra of symmetries of $(E|_{S},Q,T[1]X)$ for any background solution $S$. Namely, let $Y$, $\gh{Y}=-1$ be a $\pi$ projectable vector field on $E$. It follows, gauge symmetry $\commut{Q}{Y}$ is projectable as well and hence defines the gauge transformation of sections of $(B,\gamma,T[1]X)$. Those gauge transformations whose projections  preserve a fixed background solution $S:T[1]X\to B$ define global symmetries of $E|_{S}$.  This is of course the gPDE reformulation of the standard fact that global symmetries arise as those gauge symmetries of the system coupled to the background fields that preserve a given background solution.

The other way around, let $(E_0,Q_0,T[1]X)$ be a gPDE equipped with an action of a Lie algebra $\algg$, i.e. we are given with the ghost degree $0$ vertical vector fields $V_\alpha$ defined on $E$  and such that $\commut{Q}{V_\alpha}=0$ and $V_\alpha$ are fundamental vector fields associated to the basis elements $e_\alpha$ of $\algg$. This immediately gives rise to a gPDE  over background $(E,Q) \rightarrow (B,\gamma)$, where $E = E_0\times \algg[1]$ and the total $Q$ structure is $Q=Q_0+\dg$, where $\dg$ is a CE differential~\eqref{CEdiff} of $\algg$ with coeficients in $\cC^\infty(E_0)$. It is easy to see that $E$ can be seen as a bundle over $B=\algg[1]\times T[1]X$ and moreover $Q$ is projectable so that we are indeed dealing with a gPDE over background. The initial gPDE can be reconstructed by restricting  $E$ to the solution $c^\alpha=0$, where $c^\alpha$ are coordinates on $\algg[1]$. Of course, the gauge symmetries of the form $\commut{Q}{Y}$ and preserving $c^\alpha=0$ contain the initial symmetries $V_\alpha$. The constructed above gPDE over background can be regarded as the result of gauging the subalgebra $\algg$ of the algebra of symmetries of the initial gPDE $(E_0,Q_0,T[1]X)$.

\subsubsection{Example: gauging $u(1)$ symmetry}\label{sec:dirac}
The first example is rather trivial and is well known in one or another version. We start with a gPDE description of the complex spin $1/2$ field, which is  achieved by taking $E$ to be $T[1]X \times F$ where $F$ is the fibre of the equation manifold determined by $\gamma^a \d_a\psi(x)=0$, where $\gamma^a$ are gamma matrices satisfying $\gamma^a\gamma^b+\gamma^b\gamma^a=2\eta^{ab}$. Coordinates on $F$ can be taken to be totally-symmetric spin-tensors $\psi, \psi_a,\psi_{ab},\ldots$ satisfying $\gamma^a\psi_{a\ldots}=0$ and $\eta^{ab}\psi_{ab\ldots}=0$, where $\eta^{ab}$ is the Minkowski metric. The $Q$ structure is given by:
\begin{equation}
Qx^a=\theta^a, \qquad Q\psi=\theta^a\psi_a\,, \quad Q\psi_a = \theta^b\psi_{ab}%
\qquad \ldots \,,
\end{equation}
where $\ldots$ denote the action of $Q$ on higher jets of $\psi$ whose explicit form is not relevant in this context. The familiar $u(1)$-symmetry is determined by the following vector field:
\begin{equation}
V\psi= -i\psi\,, \qquad \commut{Q}{V}=0\,.
\end{equation}
Note that the action of $V$ on $\psi_a,\psi_{ab},\ldots $ is determined by $\commut{Q}{V}=0$. 

Gauging the $u(1)$-symmetry using the procedure described in the previous Section, results in the following gPDE over background $(E^\prime,Q^\prime) \to (B,\gamma)$:
\begin{equation}
E^\prime=F\times u(1)[1] \times T[1]X\qquad Q^\prime = Q-iC(\psi\dl{\psi}+\psi_a\dl{\psi_a}+\ldots)\,,
\end{equation}
where $C$, $\gh{C}=1$ is the new ghost variable introduced to gauge the $u(1)$ symmetry.
At the same time, $(B,\gamma)$ is given by
\begin{equation}
B= u(1)[1] \times T[1]X\,, \qquad \gamma =\dx \,.
\end{equation}
Solutions to $B$ are flat $u(1)$ connections $A=A_b(x) \theta^b$. Having fixed the background solution $A(x,\theta)=A_b(x)\theta^b$ the equations for
$\psi$ take the following form:
\begin{equation}
\dx \psi=\theta^a(\psi_a - iA_a\psi)\,, \qquad \dx \psi_a =\theta^b(\psi_{ab} - iA_b\psi_a)\,, \qquad \ldots    
\end{equation}
Taking a $\gamma$-trace of the first equation and using $\gamma^a\psi_a=0$ implies:
\begin{equation}
\gamma^a(\d_a+iA_a)\psi=0\,,
\end{equation}
i.e. a massless Dirac equation on the electro-magnetic background described by $A_b(x)$. Strictly speaking the above construction only describes pure gauge fields. However, it is clear that the consistency stays even if $A$ is not pure gauge. We discuss a systematic way to include nontrivial configurations for $A$ in Subsection~\bref{sec:weak-dirac}. 

\subsubsection{Parametrised systems}

Our next example is the so-called parametrised system. It is well known that a mechanical system can be made time-reparametrisation invariant by treating time as a phase-space variable and introducing a new evolution parameter, together with an additional gauge invariance, in such a way that reparametrisations are among the gauge transformations; see e.g.~\cite{HT-book}. This extends to field theory in a straightforward way, at least at the level of equations of motion. The parametrisation can be formulated in the full generality using a version of the gPDE formalism~\cite{Barnich:2010sw,Grigoriev:2010ic}. Here we present a more invariant and geometrical exposition of the procedure and stress that despite being a genuine gauge PDE,  the resulting parametrised system is naturally interpreted as the gPDE over background.

Let $(E,Q,T[1]Y)$ be a gPDE. Consider $(B,\gamma,T[1]X)$ given by
\begin{equation}
\pi_{T[1]X}:(T[1]X,\dx)\times (T[1]Y,\mathrm{d_{Y}}) \to (T[1]X,\dx)\,, \qquad \gamma=\dx+\dy \,,
\end{equation}
where $X$ is another copy of $Y$.\footnote{An interesting possibility is to take $X$ to be a different manifold, even of different dimension. This possibility was discussed in~\cite{Bekaert:2012vt} in the context of so-called ambient space formulations.} It is clear that $(B,\gamma,T[1]X)$ is a gPDE. Gauge PDE $(E^{\mathrm{par}},Q,T[1]X)$ describing the parametrised system is given by the pullback of $E \to T[1]Y$ by the canonical projection $\pi_{T[1]Y}:B\to T[1]Y$ to the second factor. Although $(E^{\mathrm{par}},Q,T[1]X)$ is a genuine gPDE, it can be equally well considered as a gPDE over background $(B,\gamma,T[1]X)$. Moreover, the diagonal in $T[1]X \times T[1]Y$ is obviously a solution to $B$ and the restriction of $(E^{\mathrm{par}},Q,T[1]X)$ to the diagonal coincides with the initial gPDE. At the same time, seen as a gPDE over $(T[1]X,\dx)$, gPDE $(E^{\mathrm{par}},Q,T[1]X)$ is globally trivial as a $Q$-bundle and can be also regarded~\cite{Barnich:2010sw} as a non-Lagrangian AKSZ sigma model. Moreover, if $(E,Q,T[1]X)$ is diffeomorphism invariant from the outset, this procedure gives an equivalent gPDE formulation of the initial system, at least locally. In this case one can check that $y^a,\xi^a$ seen as fibre coordinates of $E^{\mathrm{par}}$ are contractible pairs and can be eliminated, see~\cite{Grigoriev:2019ojp,Barnich:2010sw} for further details.

Let us discuss $(B,\gamma,T[1]X)$ in some more details. First of all, it is easy to check that there is one to one correspondence between solutions to $B$ and diffeomorphsims $X\to Y$ (recall that $Y$ is just another copy of $X$). Indeed, being degree preserving, any solution $\sigma_B$ defines a map $X\to Y$. At the same time, the action of $\sigma_B^*$ on the fibre coordinates on $T[1]Y$ is determined by $\sigma_B^*(\dy f)= \dx \sigma_B^* f$. To give the explicit coordinate description let $x^\mu$ and $y^a$ be the local coordinates on $X$ and $Y$, respectively. The induced coordinate system on $T[1]X$ is then $x^\mu,\theta^\mu\equiv \dx x^\mu$ while $y^a,\xi^a\equiv \dy y^a$ are coordinates on $T[1]Y$. The product $Q$ structure on $B=T[1]X\times T[1]Y$ reads as:
\begin{equation}
\gamma=\theta^\mu \dl{x^\mu}+\xi^a \dl{y^a}\,.
\end{equation}
 A solution $\sigma_B$ to $B$ is determined by functions:  $\bar y^a(x)=\sigma_B^*(y^a)$ and as we saw above $\sigma_B^*(\xi^a)=\sigma_B^*(\dy y^a)=\ddl{\bar y^a}{x^\mu}\theta^\mu$. In other words, $\sigma_B$ is entirely determined by the map $X\to Y$ and $\sigma^*_B$ is just the usual pullback map $\bigwedge^\bullet(Y) \to \bigwedge^\bullet(X)$ provided we resort to the language of the differential forms.  
 
 Let us also look at gauge transformations for $B$. A vertical gauge parameter $Z$ has the form $Z=\epsilon^a(x,y)\dl{\xi^a}$. It is enough to take $\epsilon^a$ to be $y$-independent. The corresponding gauge transformations of ``fields'' $\sigma^*_B(y^a),\sigma^*_B(\xi^a)$ are given by:
\begin{equation}
\delta_Z \sigma^*_B y^a\equiv \sigma^*_B\commut{Q}{Z}y^a=\epsilon^a(x)\,,
\quad 
\delta_Z \sigma^*_B \xi^a\equiv \sigma^*_B\commut{Q}{Z}\xi^a=\theta^\mu \ddl{\epsilon^a(x)}{x^\mu}
\end{equation}
and have a clear geometrical meaning.

It is instructive to give an explicit coordinate description of $E^{\mathrm{par}}$. Using coordinates $x^\mu,\theta^\mu,y^a,\xi^a,\psi^B$ introduced above, the explicit expression for the $Q$-structure of $E^{\mathrm{par}}$ reads as:
\begin{equation}
Q^{\mathrm{par}}=\theta^\mu\dl{x^\mu}+\xi^a\dl{y^a}+Q^B(\psi,y,\xi)\dl{\psi^B}\,.
\end{equation}
Given a background solution $\sigma_B^*(y^a)=\bar y^a(x)$, $\sigma_B^*(\xi^a)=\ddl{\bar y^a}{x^\mu}\theta^\mu$ and using $x^\mu,\theta^\mu,\psi^A$ as the coordinates on $E^{\mathrm{par}}$ pulled back to $\sigma_B$, the resulting $Q$-structure (i.e. the restriction of $Q^{\mathrm{par}}$ to $E^{\mathrm{par}}$ restricted to $\sigma_B$) reads as
\begin{equation}
Q^\prime=\theta^\mu\dl{x^\mu}+Q^B(\psi,\bar y^a(x),\ddl{\bar y^a}{x^\mu}\theta^\mu)\dl{\psi^B}\,.
\end{equation}
Of course, $Q^\prime$ is just an initial $Q$ written in a different coordinate system. This confirms that we are indeed dealing with parametrised systems. Moreover, the infinitesimal reparametrisations are now among the gauge symmetries of $E^{\mathrm{par}}$. Indeed gauge parameter $Z=\epsilon^a(x)\dl{\xi^a}$ on $B$ can be lifted to $E^{\mathrm{par}} \to B$ and define the action of infinitesimal reparametrisations on fields $\sigma^*\psi^A$.  Of course, such a lift is generally not unique and is not canonical. The reason is that reparametrisations could generally mix with the intrinsic gauge transformations of $E$.

\subsubsection{Example: Fedosov equations}

Another example of a gPDE over background comes from the Fedosov quantization of symplectic manifolds~\cite{Fedosov:1994}. It turns out that the equations determining Fedosov connection and their gauge symmetries can be naturally interpreted as equations of motion and gauge symmetries of a certain gPDE. To see this let us start with the simplified setup where the tangent bundle over the base manifold $X$ is trivial. As $E \to T[1]X$ we take $(W[1] \times \fR[2]) \times T[1]X$ equipped with the product $Q$-structure $q+\dx$, where $W$ is the Weyl algebra of a symplectic space $\fR^{n}$, with $n=2m=\dim{(X)}$ and the Moyal-Weyl star-product $W\tensor W \to W$ denoted by $*$.
The $Q$ structure of $W[1] \times \fR[2]$ is given by
\begin{equation}
q\Psi=-\frac{1}{2\hbar}\qcommut{\Psi}{\Psi}+\omega\tensor 1\,, \qquad q\omega=0\,,
\end{equation}
where $\omega$ is a coordinate on $\fR[2]$, $\gh{\omega}=2$ and $\qcommut{A}{B}=A*B-(-1)^{\p{A}\p{B}}B*A$. Here, $\Psi$ is a canonical element of $\cC^\infty(W[1])\tensor W$ given by $\psi^A\tensor e_A$, where $e_A$ is a basis  of $W$ and $\psi^A$ coordinates on $W[1]$ corresponding to this basis. Of course, $\Psi$ is just a convenient way to pack the coordinates on $W[1]$ into a single object. We also choose to work over formal series in variable $\hbar$ so that all linear operations are assumed linear over $\fC[[\hbar]]$.

A section $\sigma:T[1]X \to (W[1]\times \fR[2])\times T[1]X$ is parametrized by a 1-form $A=\theta^\mu A^B_\mu(x) e_B=\sigma^*(\Psi)$ with values in $W$ and a 2-form $\bar\omega=\sigma^*(\omega)$. The equation of motion read as:
\begin{equation}
\dx A+\frac{1}{2\hbar}\qcommut{A}{A}=\bar\omega\,, \qquad \dx \bar\omega=0\,.
\end{equation}
If we set $\bar\omega$ to be a fixed symplectic form on $X$ these are precisely the Fedosov equations determining an abelian connection. Strictly speaking, in Fedosov construction we usually assume that $A=A_0+\tilde A$, where $A_0$ is a fixed connection belonging to the quadratic subalgebra of $W$ and containing a nondegenerate soldering form.

It is clear that this system has a natural interpretation as a gPDE over background. More precisely, as a background gPDE $(B,\gamma)$ one takes $(\fR[2]\times T[1]X,\dx)$ so that background solutions are precisely closed 2-forms on $X$. It is clear that $(E,q+\dx)$ is a $Q$-bundle over $(B,\gamma)$. Note that $\omega$
can be understood as a background field associated to the degree $-1$ symmetry $\dl{\psi^0}$, where $\psi^0$
is the coordinate on $W[1]$ associated to the unit element $1\in W$. 

The system admits a number of natural generalisations. For instance, one can incorporate the bare connection $A_0$ as an additional background field. In this case it is natural to take the background gPDE to be weak so that equations of motion only require  $\bar\omega$ to be closed and $A_0$ to be torsion-free and compatible with $\bar\omega$.  In this way the background gPDE describes generic Fedosov geometry, i.e. the symplectic structure and the compatible symmetric connection.

Another natural generalisation is to extend the fibre by the additional factor $W$ so that the fibre is now $W[1]\times \fR[2] \times W$ and take
the following $Q$ structure:
\begin{equation}
q\Psi=-\frac{1}{2\hbar}\qcommut{\Psi}{\Psi}+\omega\tensor 1\,, \qquad q\Phi=-\frac{1}{\hbar}\qcommut{\Psi}{\Phi}\,,
\end{equation}
where $\Phi=\phi^A\tensor e_A$ is a generating function for coordinates on $W$. This extension results in a new 0-form field $\bar\Phi$  whose equations of motion are $d\bar\Phi+\frac{1}{\hbar}\qcommut{A}{\bar \Phi}=0$. This gives the full Fedosov system which also involves the covariant-constancy equation for an observable $\bar \Phi$.

More interesting situation is to tensor $W$ with a differential graded associative algebra. This corresponds to Fedosov quantization of 1st class constrained systems. The respective connection can be then interpreted as a version of Quillen connection. In the case of cotangent bundles this generalisation was studied in~\cite{Grigoriev:2006tt}, see also~\cite{Grigoriev:2023lcc}.
Finally, let us note that if $X$ is not parallelizable one should replace $T[1]X\times W[1]$ with a locally trivial $Q$-bundle over $T[1]X$, which  trivializes to $(T[1]U,\dx) \times (W[1],q)$ over a coordinate patch $U\subset X$.

\subsection{Linearized equations as gPDEs over background} \label{subsec:Linearized_systems}
Suppose we are given a gPDE $(E,Q,T[1]X)$. There is a canonical gPDE over background associated to it, which corresponds to the linearized version of the system. 

Given $E \to T[1]X$ there is a canonical vertical tangent bundle $VE \rightarrow E$ and a canonical lift $\Tilde{Q}$ of $Q$ to $VE$ such that $\tilde Q$ projects to $Q$ by $VE \to E$. Indeed, define $\hat Q$ on $TE$ by taking\footnote{Such a lift is sometimes called \textit{a tangent lift} or \textit{a complete lift.}}
\begin{equation}
\hat Q p^* f=p^* Q f\,,\qquad  \hat Q p^*(\tilde d f)= \tilde d p^*(Q f)\,, \quad f \in \cC^\infty(E)\,,
\end{equation}
where we employed the identification of covectors on $E$ and functions on $TE$ that are linear on fibres and $\tilde d$ denotes the differential acting from functions on $E$ to covector fields on $E$ seen as symmetric covariant tensors rather then differential forms. It is clear that $\hat Q$ projects to $Q$ on $E$ and $\hat Q$ is nilpotent. Moreover, $\hat Q$ is tangent to $VE$ seen as a submanifold in $TE$, resulting in a $Q$-bundle $(VE,\tilde Q)$ over $(E,Q)$, where $\tilde Q$ is the restriction of $\hat Q$ to $VE \subset TE$.

Let us give an coordinate expression for $\tilde Q$. Let  $(x^a, \theta^a, \psi^A)$ be local coordinates on $E$ such that $x^a,\theta^a$ are adapted coordinates on $T[1]X$ pulled back to $E$, then $Q$ has the following form:
\begin{equation}
    Q = \theta^a \ddl{}{x^a} + q^A(\psi,x,\theta)\ddl{}{\psi^A}\,.
\end{equation}
The induced coordinates on $VE$ are $(x^a, \theta^a, \psi^A, \phi^A)$ and $\tilde Q$ takes the form:
\begin{equation}
    \Tilde{Q} = \theta^a \ddl{}{x^a} + q^A(\psi,x,\theta)\ddl{}{\psi^A} + \phi^B \ddl{q^A}{\psi^B}\ddl{}{\phi^A}\,.
\end{equation}
Indeed, $\tilde Q \phi^A=\tilde dq^A=\tilde d\psi^B\ddl{q^A}{\psi^B}=\phi^B\ddl{q^A}{\psi^B}$. 

It is clear that we have arrived at a gPDE over background. The background gPDE is the starting point gPDE and the whole system describes the linearisation of the system around arbitrary background configurations. In the case where $(E,Q,T[1]X)$ is a usual PDE, this construction coincides with the notion of a tangent bundle of a PDE, see e.g.~\cite{Krasil'shchik:2010ij} for further details. 

\subsubsection{Example: zero curvature equation}
    
The gPDE description of the zero curvature equation amounts to taking $E=T[1]X\times \algg[1]$ with $Q=\dx+\dg^0$, where
$\dg^0$ is the CE differential of the Lie algebra $\algg$. The linearization of $(E,Q,T[1]X)$ is given by $VE=T(\algg[1])\times T[1]X$ and $\tilde Q=\dx+\tilde \dg$ with
\begin{equation}
\tilde \dg c=-\half\commut{c}{c}\,, \qquad  \tilde \dg \phi= - \commut{\phi}{c}\,,
\end{equation}
where coordinates on $\algg[1]$ and the fibres of $T\algg[1]$ are encoded in $\algg$-valued $c$
and $\phi$, respectively. Of course, $\tilde\dg$ can be seen as the CE differential of $\algg$ with coefficients in functions on $\algg[1]$ seen as the adjoint representation. Once a solution to the background equation is fixed the linearised system is just a covariant constancy equation for a 1-form in the adjoint representation. 

\subsection{Homogeneous gPDEs over background}\label{sec:homgPDE}
There is a simple, yet rich class of examples of a particular type, where gauge PDEs originate from so-called homogeneous bundles and are invariant under the respective group action.  Let us first recall the notion of homogeneous bundles. Let the base space $X$ be a homogeneous space $G/H$ so that there is a canonical transitive action of $G$ on $X$. Let $E_0$ be a fibre bundle over $X$. Bundle $E_0$ is called homogeneous if $E_0$ is a $G$-space and moreover $G$ acts by bundle morphisms (i.e. preserving fibres) and its induced action on the base is the canonical $G$-action on $X$. It is clear that a typical fibre is an $H$-space. It is a well-known fact that a homogeneous bundle over $G/H$ is entirely determined by its typical fibre $F$. 

More precisely, let $F$ be an $H$-space, i.e. a manifold with the $H$-action $H\times F\to F$. Consider $E_0=G\times_H F$, which is the quotient of $G\times F$ by the following equivalence relation $(g,f)\sim (gh,h^{-1}f)$, $h\in H$. At the same time the left $G$-action on $E_0$ can be defined on representatives via $g^\prime(g,f)=(g^\prime g,f)$ and is well-defined on the equivalence classes. Locally, we can trivialize $E_0$ as $(x,f)$, where $x \in X, f\in F$ by choosing a local section $\sigma:X \to G$ so that $(x,f)$ corresponds to the equivalence class $(\sigma(X),f)$. At the infinitesimal level, we are given with fundamental vector fields $\rho(a), a\in \algg=Lie(G)$ on $E_0$. These vector fields project to the fundamental vector fields on $G/H$ and satisfy $\commut{\rho(a)}{\rho(b)}=\rho(\commut{a}{b})$ (here and below we employ conventions such that $\rho$ is a homomorphisms, i.e. $\rho(a)$ at $([g],f)$  is the vector tangent to $exp(-ta)([g],f)$ at $t=0$). Further details on the geometry of homogeneous bundles can be found in e.g.~\cite{cap2009parabolic}

The typical fibre  of a generic homogeneous bundle is equipped with the $H$-action only. However, if $E \to G/H$ is additionally endowed with a flat $G$-invariant connection the typical fibre naturally inherits the action of a Lie algebra $\algg=Lie(G)$. Indeed, given such a connection, the fundamental vector field $\rho(a)$, $a \in \algg$ splits into its horizontal and vertical parts $\rho(a)=\rho_h(a)+\rho_v(a)$ so that $\rho_h(a)$ belongs to the horizontal distribution determined by the connection. Moreover, because the distribution is involutive and $\rho(\algg)$-invariant, one finds that it is also $\rho_v(\algg)$-invariant. This also implies that $\rho_v$ defines the action of $\algg$ on any fibre and hence a typical fibre is a $\algg$-space. It is important to stress that although
$E_0$ is a $G$-space by assumption, the fibre at say $[e]$ is equipped with the action of the Lie algebra of $G$ only and generally not the $G$-action, even locally. In what follows we denote by $\rho_F(a)$ the action of the fundamental vector field of $a\in \algg$ on the typical fibre $F$ which we identify as the fibre over $[e]$ (i.e. $\rho_F$ is a restriction of $\rho_v$ to $F$). Because we are dealing with a homogeneous bundle all its structures are fully determined by a fibre at $[e]$. In particular, under some technical conditions the horizontal distribution of a $G$-invariant flat connection can be generated by the vector fields of the following form:
\begin{equation}
\label{MC-rep}
\dl{x^\mu}+\omega_\mu^\alpha \rho_F(t_\alpha)
\end{equation}
where we employed a local  trivialisation $([g],f)$ and  $\omega^\alpha\equiv dx^\mu\omega_\mu^\alpha$ are coefficients of the Cartan connection 1-form on $G \to G/H$ with respect to basis $t_\alpha$ of $\algg$ and coordinate system $x^\mu$ on $X$. More precisely, $\omega^\alpha$ can be obtained as the pullback of the canonical left-invariant Maurer-Cartan (MC) form on $G$  by a trivialisation section $\sigma:X\to G$.
Here, we refrain from discussing the exact conditions under which a generic $G$-invariant flat connection on a homogenous bundle can be represented as~\eqref{MC-rep} and simply assume this in what follows. A thorough study of invariant connections on homogeneous bundles can be found in~\cite{cap2009parabolic} where, in particular, the proof of the existence of the representation~\eqref{MC-rep} is given in the case of linear connections.

In applications to gauge systems we need to reformulate a flat $G$-invariant connection in terms of $Q$-structures. To this end it is useful to describe the connection in terms of $E_0$ represented as a quotient of $G\times F$.
Let $l_\alpha\equiv l(t_\alpha)$ denote the basis left-invariant vector field on $G$ and we adopt the convention that $l(\commut{t_\alpha}{t_\beta})=\commut{l(t_\alpha)}{l(t_\beta)}$ and at $e\in G$ one has $r_\alpha=-l_\alpha$. Then the right action of $H$ on $G\times F$ is generated by $l_i+\rho_F(t_i)$, where $t_i$ denote the basis in $\algh\subset \algg$. The horizontal distribution of the canonical left-invariant flat connection on $G\times F$ can be generated by vector fields $l_\alpha+\rho_F(t_\alpha)$. It is clearly involutive and invariant with respect to the left action of $G$ on $G \times F$ (recall that $\bar g(g,f)\equiv (\bar gg,f)$). The homogeneous bundle $E_0$ is obtained by taking a quotient with respect to the right $H$-action or, equivalently,  by the distribution $l_i-\rho_F(t_i)$. It is clear that horizontal distribution $l_\alpha+\rho_F(t_\alpha)$ descends to that on the quotient, giving a canonical flat connection on $E_0$. Given a local trivialisation $\sigma:X\to G$, the corresponding connection 1-form can be obtained as a pullback of the MC form to $X\subset G$ so that in local coordinates the connection indeed takes the form~\eqref{MC-rep}.

In particular, the above discussion applies to so-called homogeneous PDEs which can be defined as homogeneous fibre bundles (typically with infinite-dimensional fibre) equipped with the invariant flat connection whose horizontal distribution is known as Cartan distribution. Note that this distribution is typically not integrable and hence one only gets $\algg\equiv Lie(G)$-action on the fibres. Nevertheless, the Cartan distribution can be still represented in the form~\eqref{MC-rep} under some technical conditions.

It turns out that the concept of homogeneous PDEs extends to the case of gauge PDEs in a straightforward way. In this discussion we limit ourselves to so-called standard gauge PDEs. Recall that a gPDE $(E,Q,T[1]X)$ is called standard if it is bigraded by the base-space degree (i.e. form degree given by homogeneity in $\theta$) and the fibre ghost degree so that $Q$ globally splits as $Q=\nabla +q$ into the form degree $0$ piece $q$ which is vertical and form 
degree $1$ piece $\nabla$.  In this case the transition functions of the underlying bundle $E\to T[1]X$ can not involve $\theta$  and hence 
$E$ is a pullback of some $E_0\to X$ by the canonical projection $T[1]X \to X$. In particular, $\nabla$ defines a connection in $E_0 \to X$. If $x^\mu,\theta^\mu,\psi^A$ are adapted local coordinates on $E \to T[1]X$, such $Q$ can be written as:
\begin{equation}
\label{standard-Q-exp}
Q=\theta^\mu(\dl{x^\mu}+\Gamma_\mu^A(x,\psi)\dl{\psi^A})+q^A(x,\psi)\dl{\psi^A}\,.
\end{equation}

It is clear that standard gPDEs can be defined in terms of bundles over $X$. Namely, let a fibre bundle $E_0\to X$ be a PDE, i.e. a bundle equipped with Cartan distribution (= flat connection), but we still allow fibre to be a graded manifold. Let in addition $E_0$ be equipped with a nilpotent  evolutionary vector field $s$  of ghost degree $1$. Horizontal forms on $E_0$ can be identified with functions on $E_0$ pulled back to $T[1]X$ by the canonical projection $T[1]X\to X$. In so doing $L_s+\dh$ is identified with the $Q$-structure on $E \to T[1]X$. It follows, in a suitable trivialisation the $Q$ structure takes precisely the form~\eqref{standard-Q-exp}, with $q=L_s$. If $E_0 \to X$ is a jet-bundle then $s$ can be understood as the BV-BRST differential, for more details see e.g.~\cite{Barnich:2000zw},  and the construction we just recalled reformulates a given local BV system as a gauge PDE, see~\cite{Barnich:2010sw,Grigoriev:2019ojp, Grigoriev:2022zlq}.  

We call a standard gPDE homogenous if the underlying bundle $E_0 \to X$ is a homogeneous bundle over $X$ and $G$-action on $E_0$ is extended to $E$ in such a way that it projects to a canonical action of $G$ on $T[1](G/H)$ and, moreover, $Q$ is $G$-invariant. It follows, both $\nabla$ and $q$ are $G$-invariant and $q$ originates from the $\algg$-invariant vector field on the fibre over $[e]\in X$.

If we restrict to the case where the $G$-invariant connections on the underlying homogeneous bundle is of the form \eqref{MC-rep}, we have the following useful representation: 
\begin{prop}\label{homgpde}
Let $(E,Q,T[1]X)$ be a homogeneous standard gPDE. Then $(E,Q,T[1]X)$ is the quotient of the following  $Q$-manifold $(T[1]G \times F,Q)$, where $Q$ is given by
\begin{equation}
Q=\nabla+q\,, \qquad \nabla=\mathrm{d}_G+\omega^\alpha \rho_F(t_\alpha),
\end{equation}
where $\omega^\alpha$ are components of the canonical  MC form on $G$ and $q$ is a $Q$-structure on $F$.  Functions on the quotient are those annihilated by 
\begin{equation}
\label{subal}
 \hat l(t_i)+\rho_F(t_i)\,,\qquad  I_{l(t_i)}\equiv l^M_i\dl{\theta^M}\,,
\end{equation}
where $l(t_i)$ is the left-invariant vector field associated to $t_i\in \algh$ and $\hat l(t_i)$ is its natural lift to $T[1]G$.  Note that $I_{l(t_i)}$ is just a contraction with $l(t_i)$ if one identifies forms on $G$ with functions on $T[1]G$.
\end{prop}
\begin{proof}
Notice the following relation:
\begin{equation}
\nabla I_{l(t_i)}+I_{l(t_i)} \nabla  =\hat l(t_i)+\rho_F(t_i)
\end{equation}
which holds because $\omega^{\alpha}(l_\beta)=\delta^\alpha_\beta$ by definition of the left-invariant MC form on $G$. It follows, $\nabla$ is well defined on the subalgebra \eqref{subal} and hence defines the homological vector field on $E$, which encodes a $G$-invariant connection on $E_0$. Furthermore, $q$ is $\algg$-invariant by assumption and hence descends to the quotient as well, giving a gPDE
structure on $E$. It is clear that any $G$-invariant $q$ arises this way.
\end{proof}

It is natural to regard gPDEs of the form described in the above Proposition as generalizations of the usual homogeneous PDEs. Indeed, if $q=0$ and $F$ is a real linear manifold then Proposition \bref{homgpde}
amount to the standard description of the associated flat connection in the homogeneous bundle. In this case the graded geometry language is of course superfluous.

The homogeneous gPDEs of the form described in the above Proposition arise as restrictions of specific gPDEs over background to suitable background solutions. To see this let us first present a local construction by assuming the underlying bundle trivial. Consider $\Tilde{E} = F \times \mathfrak{g}[1] \times T[1]X $ seen as a bundle over $B =  \mathfrak{g}[1] \times T[1]X \rightarrow T[1]X$. $\Tilde{E}$ is a gPDE over background with $\tilde Q = \dx + \dg + q$, where 
$\dg$ is a CE differential of $\algg$ with coefficients in $\cC^\infty(F)$ and seen as a vector filed on $\algg[1]\times F$. If $c^\alpha$ are linear coordinates on $\algg[1]$ associated to basis $t_\alpha$ and $\rho_F(a)$  is a fundamental vector field of $a\in \algg$ defined on $F$, the differential reads explicitly as:
\begin{equation}
\label{Qhom}
    \tilde Q=\dx-\half c^\alpha c^\beta U_{\alpha\beta}^\gamma \dl{c^\gamma}+c^\alpha\rho_F({t_\alpha})+q \,.
\end{equation}
It clearly projects to $B$ and solutions of $B$ are flat $\algg$-connections. Moreover, the restriction of $(\tilde E,\tilde Q)$ to a solution of $B$ is the corresponding homogeneous gPDE provided the flat $\algg$-connection is a Cartan one, i.e. satisfy the nondegeneracy codition.

Let us finally present a global version of the construction. Take $T[1]G\times \algg[1]\times F$ with the total $Q$-structure being 
\begin{equation}
\tilde Q=\mathrm{d_G}+\dg+q\,.
\end{equation}
Just like in the Proposition~\bref{homgpde}, the quotient of $T[1]G\times \algg[1]\times F$ can be performed by considering functions annihilated by
\begin{equation}
\hat l_i+\commut{\dl{c^i}}{\dg}\,, \qquad I_{l_j}+\dl{c^j}\,,
\end{equation}
where $l_i\equiv l(t_i)$. It is easy to see that $\tilde Q=\mathrm{d_G}+\dg+q$ is well defined on the subspace of such functions, defining a new $Q$-bundle over $T[1]X$. Indeed, locally one can identify the subspace with functions on $T[1]X\times \algg[1]
\times F$.

In the field theory context, homogeneous gPDEs of the form described by Proposition~\bref{homgpde}  emerge, for instance,  in the gauge PDE description~\cite{Barnich:2004cr,Barnich:2006pc,Alkalaev:2008gi,Bekaert:2009fg,Alkalaev:2009vm,Bekaert:2012vt} of gauge fields on constant curvature spaces and conformal fields. In this case $\cF$ is typically the CE complex of a parabolic subalgebra $\algp\in\tilde\algg$ with coefficients in the $\algg\oplus\tilde\algg$-bimodule. More specifically, $\algg$ and $\tilde\algg$ form a version of Howe dual pair. In the case of $q=0$ this example was already known in the context of the unfolded description of free fields on coset spaces, see~\cite{Lopatin:1987hz,Shaynkman:2004vu}. Somewhat similar structures are also known in the context of equivariant BV formulation~\cite{Bonechi:2022aji,Cattaneo:2024ows}.

A homogeneous gPDE of the form described by Proposition~\bref{homgpde} has $\algg$ as a subalgebra of the algebra of symmetries. Indeed, let $\epsilon^\alpha(x)t_\alpha$ be a covariantly constant section of the bundle $G\times_H \algg$, i.e. $\dx\epsilon+\commut{\omega}{\epsilon}=0$.  It follows, the vector field $V=\epsilon^{\alpha}\rho_F(t_\alpha)$ on $E$ satisfies $\commut{Q}{V}=0$ and hence represents a symmetry of $(E,Q,T[1]X)$. It is a nontrivial symmetry unless $\rho(t_\alpha)$ is $q$-exact.

It is easy to understand the origin of this symmetry in terms of the gPDE over background constructed above. Indeed, consider a vector field $W$ defined on $\tilde E$ and given locally by
$W=\commut{\tilde Q}{\epsilon^\alpha\dl{c^\alpha}}$, i.e. $W$ is a gauge transformation. The condition that $W$ leaves the solution $S$ intact implies that $\epsilon$ is covariantly constant. The restrictions of $W$ to $\tilde E|_S$ coincides with $V$. 

To summarize, given a homogeneous space $G/H$ and a $Q$-manifold $(F,q)$ equipped with the action of $\algg$ preserving $q$ one naturally associates to this data a homogeneous gPDE using the construction of Proposition~\bref{homgpde}. More generally, given a principle $H$-bundle over a general manifold of dimension $\dim(G/H)$
one can construct a gPDE over background, whose total space is the associated bundle with the fibre $(F,q)\times\algg[1]$ (in so doing both $F$ and $\algg[1]$ are seen as $H$-spaces). However, it is not guaranteed that the background gPDE admits a flat Cartan connection as a solution unless the  starting point principle bundle admits it.

\subsubsection{Example: Fronsdal fields}
As an example let us construct a homogeneous gPDE over background for Fronsdal fileds~\cite{Fronsdal:1978rb} in Minkowski space.

Theory is linear and so is the fiber manifold $F$. It is convenient to encode the coordinates of the fiber $F$ in terms of the following generating function of auxiliary commuting variables $y^a, p^a$ of vanishing ghost degree and an anticommuting variable $c$, $\gh{c}=1$:
\begin{equation}
    \begin{gathered}
       \Psi(y,p,c)= c\Phi(y^a,p^a)+ C(y^a, p^a) , \qquad \gh{\Phi} = 0\,, \gh{C} = 1 \,.
    \end{gathered}
\end{equation}
obeying 
\begin{equation}
\begin{gathered}
    \ddl{}{y^a}\ddl{}{y_a} \Psi = \ddl{}{y^a}\ddl{}{p_a} \Psi = \ddl{}{p^a}\ddl{}{p_a} \Psi = 0\,, 
\end{gathered}
\end{equation}
These conditions mean that coordinates $\phi^{a_1 a_2...}{ }_{b_1 b_2...}, \xi^{a_1 a_2...}{ }_{b_1 b_2...}$ encoded in $\Phi,C$ are totally traceless tensors. The fibre differential $q$ is defined as 
\begin{equation}
    q\Psi = cp^a\ddl{}{y^a}\Psi\,, 
\end{equation}
or, more explicitly,
\begin{equation}
    \qquad q\Phi = p^a\ddl{}{y^a}C, \qquad qC = 0\,.
\end{equation}
Generating functions are assumed to be polynomials in $p^a$ and formal power series in $y^a$.

Manifold $F$ is equipped with a linear action of Poincar\'e algebra compatible with the fibre differential. The action of translations $T_a$ and Lorentz rotations $M_{ab}$ is determined by: 
\begin{equation}
\begin{aligned}
\rho_v(T_a)\Psi &= \ddl{}{y^a}\Psi, \\
\rho_v(M_{ab})\Psi&=(y_b\ddl{}{y^a}-y_a\ddl{}{y^b} + p_b \ddl{}{p^a}-p_a \ddl{}{p^b})\Psi\,.
\end{aligned}
\end{equation}
This defines CE differential on $\algg[1]\times F$. If $\xi^a$, $\half \rho^{ab}$ are coordinates on $\algg[1]$
corresponding to the basis $T_a,M_{ab}$ in $\algg$ the differential reads as
\begin{equation} 
\begin{gathered}
    \dg\xi^a = -\rho^{a}{ }_c\xi^c \,, \qquad 
    \dg\rho^{a}{ }_b = -\rho^{a}{ }_c \rho^{c}{ }_b\,, \\
    \dg\Psi=\xi^a \ddl{}{y^a}\Psi + \half \rho^{ab}(y_b\ddl{}{y^a}-y_a\ddl{}{y^b} + p_b \ddl{}{p^a}-p_a \ddl{}{p^b})\Psi \,.
\end{gathered}
\end{equation}

The $Q$-structure on $E = F \times \mathfrak{g}[1] \times T[1]X$ is defined as:
\begin{equation}
    \tilde Q = \dx + \dg + q\,.
\end{equation}
It is easy to see that this differential is precisely of the form \eqref{Qhom}. Reducing to a solution of the background boils down to choosing a flat Cartan connection, valued in the Poincare algebra, i.e. choosing the coframe  $e^a_\mu\theta^\mu = \sigma^*\xi^a$ and the Lorentz connection $\omega^{ab}_{\mu}\theta^\mu = \sigma^*\rho^{ab}$. The resulting system describes free higher-spin fields propagating in a flat background. This formulation of Fronsdal theory was originally put forward in~\cite{Barnich:2004cr}, see also~\cite{Alkalaev:2008gi}.

Finally, let us mention that one can make manifest an additional structure here. Namely, $F$ can be represented as $F_0\times \tilde\algg[1]$, where coordinates on $F_0$
are coefficients of polynomials in $y,p$ and $\tilde\algg$  is a 1-dimensional subalgebra of $sp(2)$ algebra.
In the standard basis the $sp(2)$ action on $F_0$ is given by:
\begin{equation}
p^a\dl{y^a}\,, \quad p^a\dl{p^a}-y^a\dl{y^{a}}\,, \quad y^a\dl{p^{a}}\,,
\end{equation}
and $\tilde\algg$ is generated by $p^a\dl{y^a}$. If we identify the linear coordinate on $\tilde\algg[1]$ with $c$, vector field $q$ can now be seen as a CE differential of $\tilde\algg$ with coefficients in functions on $F_0$ so that
$d_\algg+q$  is a CE differential of $\algg\oplus \tilde\algg$ with the same coefficients.  Note that together with Lorentz subalgebra of $\algg$ this $sp(2)$ makes $F_0$ into a bimodule where these two algebras commute. This structure is the example of so-called the reductive dual pair correspondence~\cite{Howe}. Further examples of linear homogeneous gPDEs arising in this way can be found in~\cite{Barnich:2006pc,Alkalaev:2008gi,Bekaert:2009fg,Alkalaev:2009vm}.

\subsection{Generalisation: weak gPDEs over background}

The notion of gPDE over background has a natural generalization where both gPDEs are allowed to be weak.

We first need to recall the notion of a projectable distribution. Given a bundle $E \xrightarrow[]{{\pi_B}} B$, a distribution $\cK$ on $E$ is called projectable if it is generated by projectable vector fields. It follows that a projectable distribution determines the distribution on $B$, denoted by $\cK_B={\pi_B}_*\cK$. Note that any vector field $v$ from ${\pi_B}_*\cK$ can be lifted to a vector field $V \in \cK$ such that ${\pi_B}_*V=v$. Indeed, this can be done by representing $v$ as a linear combination of vector fields of the form ${\pi_B}_* W$, $W\in\cK$.
    

Now we define:
\begin{definition}
A weak gPDE over background $(E,Q,\cK)\rightarrow (B,\gamma, T[1]X)$ is a pair of almost-$Q$ bundles  ${\pi_B}:(E,Q)\to (B,\gamma)$ and $\pi_{T[1]X}:(B,\gamma)\to (T[1]X,\dx)$ such that $E$ is equipped with an involutive and ${\pi_B}$-projectable distribution $\cK$ satisfying $L_Q\cK \subset \cK$, $Q^2\in\cK$ and $(\pi_{T[1]X}{} \circ {\pi_B})_*  \cK = 0$.
\end{definition}
Note that $L_Q\cK \subset \cK$, $Q^2\in\cK$ implies $L_\gamma \cK_B\subset \cK_B$, $\gamma^2 \in \cK_B$, where $\cK_B={\pi_B}_*\cK$ because ${\pi_B}_*Q=\gamma$ which ensures that $(B,\gamma, \cK_B,T[1]X)$ is itself a weak gPDE. Note also that $(E,Q,\cK,T[1]X)$ seen as a bundle over $T[1]X$ with the projection being $\pi_{T[1]X}\circ {\pi_B}$ is a weak gPDE and
that every solution of $(E,Q,\cK,T[1]X)$ projects to a solution of $(B,\gamma,\cK_B,T[1]X)$
Indeed, $\cK$ is $\pi_{T[1]X}\circ {\pi_B}$-vertical and hence projects to a trivial distribution on $T[1]X$. This gives an alternative definition of weak gPDE over background as a bundle ${\pi_B}:E\to B$ in the category of bundles over $T[1]X$ such that $E$ is a weak gPDE over $T[1]X$ and both $Q$ and $\cK$ are ${\pi_B}$-projectable.

Let us remark that the linearization procedure introduced in Section~\bref{subsec:Linearized_systems} can be easily generalized to the case of weak gPDEs. Namely, given a weak gPDE $(E,Q,\cK, T[1]X)$ one can construct $(VE,\Tilde{Q}, \Tilde{\cK}) \rightarrow (E,Q,T[1]X)$ where $\Tilde{\cK}$ is the complete lift of $\cK$.

\subsubsection{Example: Dirac equation in EM background} \label{sec:weak-dirac}

The generalization of the example presented in Subsection~\bref{sec:dirac} to not necessarily flat $u(1)$ connection is immediate if we consider $(B,\gamma, T[1]X)$ as a weak gPDE, with the $\gamma$-invariant distribution $\cK_B$ generated by vector fields $\theta^a\theta^b\dl{C}$. 

More precisely, let us consider a weak gPDE over background $(E,Q,\cK)\rightarrow (B,\gamma, T[1]X)$. The total space is taken to be a linear graded manifold with coordinates: 
\begin{equation}
\begin{gathered}
    \psi \quad \gh{\psi} = 0\,, \qquad 
    C, \quad \gh{C} = 1 \,,\qquad  x^a, \theta^a
\end{gathered}
\end{equation}
and the projection $\pi_B:E \to B$ is given by $\pi_B^* (C)=C$.  The $Q$ structure is given by 
\begin{equation}
    Q\psi = -iC\psi - i\frac{1}{4} m\theta^a \gamma_a \psi, \qquad QC = \gamma C = 0
\end{equation}
and $\cK$, seen as a module module over $\cC^\infty(E)$, is generated by:
\begin{equation}
\begin{gathered}
    (\gamma_{(a}\theta_{b)} - \frac{1}{4} \eta_{ab}\gamma^c \theta_c)\ddl{}{\psi}, \qquad  \theta^a \theta^b \ddl{}{\psi}, \qquad   
    \theta^a \theta^b \ddl{}{C} \,.
\end{gathered}
\end{equation}
It is easy to see that $[Q,\cK] \in \cK$ and $\cK$ is $\pi_B$-projectable with $\cK_B=(\pi_B)_*\cK$ generated by $\theta^a \theta^b \ddl{}{C}$.

It is easy to see that a solution $S$ to the background gPDE $(B,\gamma)$ describes a general $u(1)$-connection $S^*(C)A=A_a(x)\theta^a$ subject to the usual gauge transformations. At the same time, the pullback of $E$ to $S$ describes the following equations of motion for $\psi$: 
\begin{equation}
    \theta^a\partial_a\psi + i \theta^a A_a \psi + \frac{1}{4}im\theta^a \gamma_a \psi+ \alpha_{ab}(x)\theta^a\gamma^b = 0\,,
\end{equation}
where $\alpha_{ab}(x)$ are components of a symmetric traceless tensor field on $X$ , which we use to parameterize an arbitrary vector from $\cK$ in "$\psi$ direction", $k_\psi = \alpha^{ab}(\gamma_{(a}\theta_{b)} - \frac{1}{4} \eta_{ab}\gamma^c \theta_c)\ddl{}{\psi}$. Contracting the equation with $i\gamma^a$ one gets
\begin{equation}
   ( i\gamma^a \partial_a - \gamma^a A_a - m)\psi = 0\,,
\end{equation}
which is the Dirac equation in generic EM background. Of course, gauge symmetries preserving a fixed background solution is the starting point $u(1)$ symmetry.

\section{Presymplectic gPDEs over backgrounds} \label{sec:presymp-bgpde}

We are mostly interested in Lagrangian theories and hence need a notion of gPDE over background such that setting background to a particular solution to $(B,\gamma,\cK_B,T[1]X)$ one obtains not only equations of motion but also the respective Lagrangian.

\begin{definition} \label{def:B-presymp}
A presymplectic gPDE over background is an almost-$Q$ bundle  $(E,Q) \xrightarrow{{\pi_B}} (B, \gamma)$ over a 
weak gPDE $(B,\gamma, \cK_B,T[1]X)$, equipped with a presymplectic structure $\omega \in \bigwedge^2(E)$, $deg(\omega) = n-1$ and $\cL \in \cC^{\infty}(E)$, $\gh{\cL}=n$ such that
    \begin{equation}\label{compat2_theory}
        d\omega = 0, \qquad i_Q\omega + d\cL \in {\cI_B}, \qquad \frac{1}{2}i_Q i_Q \omega + Q\cL = 0\,,
    \end{equation}
    where $\cI_B$ is an ideal generated by forms ${\pi_B}^*\alpha$, $ \alpha \in \bigwedge^{1}(B)$, $i_{\cK_B} \alpha = 0$. Moreover, we require that the distribution $\cK$ generated by all projectable $k$ obeying $i_k \omega \in \cI_B\,, \pi_B{}_* k \in \cK_B $ is such that $\pi_B{}_* \cK = \cK_B$.
\end{definition}
Note that the last condition implies that any $k \in \cK_B$ can be lifted to a projectable vector field satisfying $i_k \omega \in \cI_B$.
In what follows we use the notation $(E,Q,\omega)\to (B,\gamma, \cK_B,T[1]X)$ or simply $(E,Q,\omega)\to (B,\gamma, \cK_B)$ if the choice of spacetime is clear from the context. Speaking informally, the bundle $(E,Q) \xrightarrow{{\pi_B}} (B,\gamma)$ is the bundle of ``fields over background geometry'', with the background geometry reformulated as a (weak) gPDE. If we take a trivial background gPDE $(B,\gamma,\cK_B,T[1]X)=(T[1]X,\dx,0,T[1]X)$ the above definition reduces to that of a presymplectic gPDE \bref{Def:w_presymp_gPDE}. 
\begin{rem}\label{rem:alt-def-chi}
    Just like in the case of presymplectic gPDEs, see~\cite{Dneprov:2024cvt}, an alternative definition of the presymplectic gPDE over background is to require $E$ to be equipped with $Q,\chi,\cL$ such that \eqref{compat2_theory} holds for $\omega=d\chi$.
Recall that for $\gh{\omega}\neq 0$ a potential exists globally. In this version the ambiguity in the choice 
of $\chi$ is manifest. Indeed,  changing $\chi$ to $\chi+\alpha$, $\alpha \in \cI_{B}$ one finds that \eqref{compat2_theory} is again satisfied for $\cL-i_Q\alpha$. For instance, 
\begin{equation}
\iota_Q \iota_Q d\alpha=\iota_Q L_Q \alpha+\iota_Q d \iota_Q \alpha=
 L_Q \iota_Q\alpha+\iota_Q d \iota_Q \alpha=
2Q(i_Q\alpha)\,,
\end{equation}
where in the second equality we used  $i_{Q^2}\alpha=0$ which holds because $Q^2$ projects to $\gamma^2 \in \cK_B$ and $\alpha \in \cI_B$. $i_{\gamma^2}\beta=0$ for any $\beta$ such that $i_{\cK_B}\beta=0$. Moreover, the integrand of the AKSZ-like action \eqref{BV-AKSZ} determined by $\chi$ is unchanged under $\chi\to  \chi+\alpha$, $\cL\to \cL-i_Q\alpha$ provided $\pi_B\circ \sigma$ solves the background gPDE. This is immediately clear for $\alpha$ of the form $f \, \pi_B^*\beta$, where $\beta$ is the 1-form on $B$ and $f\in \cC^\infty (E)$.
\end{rem}

Let us stress that at this stage we do not assume any regularity conditions on $\omega$ but this has to be done at some later stage. The condition of the existence of ${K}$ ensures that, at least locally, a vector field from $\cK_B$ can be lifted to a vector field on $E$ that belongs to the kernel of $\omega$ modulo $\cI_B$.

The following further remarks are in order:
\begin{rem} For a background gPDE as defined in~\bref{def:B-presymp} one has $i_{Q^2}\omega \in \cI_B$. To prove this observe that:
\begin{multline}
    i_{Q^2}\omega = L_Qi_Q\omega - i_QL_Q\omega = L_Qi_Q\omega + i_Q d i_Q \omega = \\ 2L_Q i_Q\omega + d(i_Q i_Q \omega) = 2L_Q(-d\cL + \cI_B) +  d(i_Q i_Q \omega) =\\ d(i_Q i_Q \omega + 2Q\cL) + L_Q \cI_B\,.
\end{multline}
Here, the $d$-exact term vanishes thanks to the presymplectic master-equation while the last term lies in $\cI_B$ because $L_Q$ preserves $\cI_B$. Indeed, $L_Q \pi^* \alpha = \pi^*L_\gamma \alpha$ and since $[\gamma,\cK_B] \in \cK_B$ it follows that $i_{\cK_B}L_\gamma\alpha = 0$.
\end{rem}

\begin{rem} \label{rem:forget}
    \textbf{Forgetting the symplectic structure.} Given a presymplectic gPDE over background one can forget the presymplectic structure and consider this system as a weak gPDE over background $(E,Q,\cK) \to (B,\gamma,\cK_B,T[1]X)$, where $\cK$ is the distribution generated by all projectable vector fields $Y$ such that $i_Y\omega \in \cI_B$, ${\pi_B}_*Y \in \cK_B$. Note that since every vertical vector field is projectable, $\cK$ contains the whole of the vertical kernel of $\omega$. 
\end{rem}

\begin{rem}
    \textbf{When a weak gPDE over background is presymplectic?}  Given a weak gPDE over background $(E, Q, \cK) \to  (B, \gamma, \cK_B, T[1]X)$ it is natural to ask if there exists a presymplectic structure $\omega$ such that $\omega$ obeys \eqref{compat2_theory} and the distribution $\cK^\prime$ generated by vector fields $ \pi_B{}_* k \in \cK_B\,, i_k \omega \in \cI_B$ coincides with $\cK$. These condition ensure that  $(E, Q, \omega) \to (B, \gamma, \cK_B)$ is a presymplectic gPDE over background. Moreover, applying the ``forgetful functor'' described in Remark~\bref{rem:forget} one gets back the initial weak gPDE over background.
\end{rem}


\subsection{Presymplectic gPDE on background solution}
Given a particular solution, i.e. "a fixed background", of the background gPDE $B$, the pullback of $E$ to the solution merely defines a weak presymplectic gPDE over $T[1]X$. The latter is interpreted as a gauge system in the fixed background. 

\begin{theorem}
Let  $(E,Q,\omega, \cL) \to (B,\gamma, \cK_B,T[1]X)$ be a  presymplectic gPDE over background $B$ and
$S: T[1]X \rightarrow B$ be a solution to the background weak gPDE. This naturally defines a presymplectic gPDE over $T[1]X$ and hence a local BV system, provided $\omega$ is quasi-regular
on $E$ pulled back to $S$.
\end{theorem}
\begin{proof}
    Let $S$ be a fixed solution to the background, i.e.
    there exists $k_S \in \cK_B$ such that  
    \begin{equation}
    \label{solution-B}
        \dx \circ S^* -S^*\circ \gamma = S^*\circ k_S\,.
    \end{equation}
    Note that $S$ can be seen as a submanifold of $B$, and the vector field $\gamma + k_S$ is tangent to $S$. Indeed, if $b^i$ are fiber coordinates, the submanifold is the zero locus of $\pi_{T[1]X}^*(S^*(b^i)) - b^i$ and the ideal generated by these functions is preserved by $\gamma + k_S$.
    Moreover, $S$ defines an isomorphism $T[1]X$ to $S$, while
    equation~\eqref{solution-B} says that $\gamma + k_S$ restricted to $S$ coincides with $\dx$ pushed forward to $S$
    by $S$ seen as a map $T[1]X\to B$. Below we denote by $E|_S$  the pullback of $E$  to $S$.

    Although $E|_S=(\pi_B)^{-1}S$ is naturally a submanifold in $E$, vector field $Q$ is not tangent to $E|_S$ in general. However, $Q$ can be adjusted by a vector field from $\cK$ to a vector field tangent to $E|_{S}$. Indeed, according to the Definition~\bref{def:B-presymp} any vector field from $\cK_B$ can be uplifted to a projectable vector field from $\cK$. Recall, that $\forall k \in \cK$ it follows $i_k\omega\in \cI_B$. If 
$\tilde k_S \in \cK$ is a lift of $k_S\in \cK_B$, vector field $Q_S=Q+\tilde{k}_S$ on $E$ projects to $\gamma+k_S$ on $B$ and hence $Q_S$ is tangent to $E|_{S}\subset E$. It follows, $(E_S,Q_S)$ is naturally an almost-$Q$ bundle over $(T[1]X,\dx)$.

Furthermore, conditions \eqref{compat2_theory} remain intact if $Q$ is replaced by $Q_S$. Indeed, the second condition holds because $i_{{\tilde{k}_S}} \omega \in \cI_B$ while the third one gives:
\begin{multline}
        \frac{1}{2}i_{Q+{\tilde{k}_S}} i_{Q+{\tilde{k}_S}} \omega + (Q+{\tilde{k}_S})\cL = \\
        \frac{1}{2}i_Q i_Q \omega + Q \cL + i_{\tilde{k}_S}(i_Q\omega +d\cL)+ \frac{1}{2}i_{\tilde{k}_S}i_{\tilde{k}_S}\omega  
        = 0 + i_{\tilde{k}_S}\cI_B = 0\,,
\end{multline}
where the last equality holds because $i_{\tilde{k}_S} \alpha=0$ for any 1-form from $\cI_B$ (indeed, $\cI_B$ is generated by $\pi^*\beta$ with $i_{{k}_S}\beta =0$).    

Bundle $E|_{S}$ can be seen as a bundle over $T[1]X$ by taking as a projection the composition $\pi_S$ of ${\pi_B}$ restricted to $E|_{S}$ and the projection $\pi_{T[1]X}:B\to T[1]X$ restricted to $S\subset B$. Moreover, $Q_S\circ \pi_S^* = \pi_S^* \circ \dx $. Under the pullback by the inclusion map $i: E|_S \rightarrow E$ the ideal $\cI_B$ goes to the ideal $i^*\cI_B$ which is generated locally by the forms $\pi^*(dx^a), \pi^*(d\theta^a)$. Then pulling back the defining equations of \bref{def:B-presymp} one finds that $(E|_S, Q_S, i^*\omega, T[1]X)$ satisfy all the axioms of a weak presymplectic gauge PDE \bref{Def:w_presymp_gPDE} if $i^* \omega$ is quasi-regular, in which case $(E|_S, Q_S, i^*\omega, T[1]X)$ defines a local BV system \cite{Grigoriev:2022zlq, Dneprov:2024cvt}. 
\end{proof}

A simple informal way to understand the construction is to see $E$ as a bundle over $T[1]X$ (by composition of projections) and to consider the subspace of sections for which the diagram 
\begin{tikzcd}
E \arrow[dr,"\pi_B"'] & \\
& B  \\
T[1]X \arrow[uu,"\sigma"'] \arrow[ur,"S"]
\end{tikzcd}
commutes for a given $S$. Practically, this means that we set the background fields equal to some fixed functions. On this subspace of sections the differential $\gamma + k_s$ ``becomes'' $\dx$ and the axioms of \bref{def:B-presymp} imply the axioms of \bref{Def:w_presymp_gPDE}. 

The BV action of the theory is given by the BV-AKSZ-like action:
\begin{equation} \label{act_back_restr}
    S_{BV} [\hat{\sigma}] = \int_{T[1]X} \Hat{\sigma}^*(\chi)(\dx) + \hat\sigma^*(\cL)\,,
\end{equation}
where $\hat{\sigma}$ denotes a supesection of $E|_S$. Strictly speaking the usual BV formulation is obtained by taking a quotient of the space of supersections by the kernel distribution of the presymplectic structure induced by $\omega$ on the space of supersections, see Section~\bref{sec:presymplectic}.

\subsection{Capturing background symmetries}\label{sec:capturing}

The information about gauge symmetries associated with background gauge transformations are generally not captured by the BV action \eqref{act_back_restr}. Nevertheless, it turns out that these gauge transformation are encoded in the almost-$Q$ structure on $E$.  To see this let us  consider a background presymplectic gauge PDE, as defined by \bref{def:B-presymp}. By composition of projections $E$ is naturally a bundle over $T[1]X$. Note that one could try to consider $(E,Q,\omega,T[1]X)$ as a presymplectic gPDE but $\omega$ is generally $Q$-invariant modulo $\cI_B$ only so that the conditions of Definition~\bref{Def:w_presymp_gPDE} are not guaranteed.

Consider the AKSZ-like action on the space of sections $\sigma:T[1]X\to E$:
\begin{equation} \label{act_3}
    S[\sigma]= \int_{T[1]X} (\sigma^*)(\chi)(\dx) + (\sigma^*)(\cL)
\end{equation}
and restrict to sections that solve the background gPDE, namely,
\begin{equation}
    (\sigma^* \circ \pi_B^*) \circ \gamma + (\sigma^* \circ \pi_B^*) \circ k  = \dx \circ (\sigma^* \circ \pi_B^*)
\end{equation}
for some $k\in K_B$

Now we attempt the gauge tranformation~\eqref{gaugetransf2} in this more general situation:
$\delta_Y\sigma^{*}={\sigma}^{*}\circ \commut{Q}{Y}-\commut{\dx}{y}\circ \sigma^*$. As in the case of weak gPDE over background we also require $Y$ to be projectable and to preserve $\cK$. In particular, $Y$ preserves $\cI_B$. Applying this transformation to the action we get:
\begin{equation}\label{act_gauge_var}
    \delta S = \int \sigma^*(L_Y(\frac{1}{2}i_Qi_Q\omega + Q\cL)) + i_{E_\sigma}L_Y(i_Q\omega + d\cL) + \frac{1}{4}i^2_{E_\sigma} L_Y\omega\,,
\end{equation}
where $E_{\sigma} = \dx\circ\sigma^* - \sigma^* \circ Q$ is a vector field along the map $\sigma$ and $i_{E_\sigma}\alpha = \sigma^*(\alpha)(\dx) - \sigma^*(i_Q\alpha)$ for $\alpha \in \bigwedge^1(E)$. The explicit expression for $i^2_{E_\sigma}L_Y\omega$ is 
\begin{equation*}
    \half i^2_{E_\sigma}L_Y\omega = \sigma^*(L_Y\omega)(\dx, \dx) - 2\sigma^*(i_QL_Y\omega)(\dx) + \sigma^*(i_Q i_Q L_Y \omega)\,.
\end{equation*}
Details of the derivation of~\eqref{act_gauge_var} are given in the Appendix~\bref{AppA}.

It is easy to see the first term in \eqref{act_gauge_var} is zero due to conditions \eqref{compat2_theory}. Because 
$Y$ preserves $\cK$ one has $L_Y \cI_B \in \cI_B$ and  and hence $i_{E_{\sigma}}L_Y(i_Q\omega + d\cL) = i_{E_{\sigma}}\cI_B$. Furthermore, as we restricted ourselves to sections which project to solutions of the background system, we have $i_{E_{\sigma}} \cI_B = 0$. Finally, vanishing of the last term in \eqref{act_gauge_var} is an extra condition which is satisfied provided 
\begin{equation} \label{background-symm-extracond}
    L_Y \omega \in \cI_B\,.
\end{equation}

As we are going to see later, the above gauge transformations involve the gauge symmetries of $(B,\gamma,\cK_B)$, which are generally not captured by the BV formulation determined by $E$ pulled back to a fixed solution $S\subset B$.

\subsubsection{Example: scalar over  Riemmanian geometry}

We start by presenting the background weak gPDE $(B,\gamma, \cK_B)$  describing the Riemmanian geometry. For simplicity we restrict ourselves to the situation where the tangent bundle is trivial. 

The total space is $B = \mathfrak{g}[1] \times T[1]X$, where $\mathfrak{g}[1]$ is the shifted Poincare algebra parameterized by coordinates $\xi^a, \rho^{ab}$ and $X$ is the spacetime which we take four-dimensional for simplicity. The action of $\gamma$ is:
\begin{equation}
    \begin{gathered}
        \gamma \xi^a = -\rho^{a}{ }_k \xi^k\,, \\
        \gamma \rho^{ab} = -\rho^{a}{ }_k \rho^{kb}
    \end{gathered}
\end{equation}
and $\cK_B$ generated by vector field
\begin{equation} \label{K-Riemann}
\begin{gathered}
    Y_{ab}{ }^{cd} = \xi^a \xi^b \ddl{}{\rho^{cd}} \,,\\
    X^{ab}{ }_{cd} = \xi^a \xi^b \xi_c \ddl{}{\xi^d}\,.
\end{gathered}
\end{equation}
It is easy to see that it is involutive and $[Q,\cK_B] \in \cK_B$. Moreover, it is straightforward to check that the equations of motion do not restrict the curvature in the Lorentz sector $d\lorcon^{ab}+\lorcon^{a}{}_c\lorcon^{cb}$ and impose the zero torsion condition $de^a+\lorcon^a{}_{b} e^b=0$. Here, $\lorcon^{ab}=\sigma^*(\rho^{ab})$ and $e^a=\sigma^*(\xi^{a})$. More precisely, it is easy to check that $\cK_B$ does not affect the zero torsion equation arising in the $\xi$-sector, the equations in the $\rho$ sector has the form:
\begin{equation}
\dx\sigma^*(\rho^{ab}) + \sigma^*(\rho^a{}_k \rho^{kb})=\sigma^*{}R^{ef}_{cd}Y^{cd}_{ef}\rho^{ab}=
\sigma^*{}R^{ab}_{cd}\xi^c\xi^d\,,
\end{equation}
where $R^{ab}_{cd}$ denote arbitrary coefficients. Of course this equation just sets $\sigma^* R$ equal to Lorentz curvature and hence do not impose any conditions on $\lorcon^{ab}$. So we are indeed dealing with the weak gPDE description of the Riemannian geometry.

Now let us consider the following presymplectic gPDE over background $(E,Q, \omega) \rightarrow (B,\gamma)$ where $E = F \times B$ and $F$ is parameterized by coordinates $\phi, \phi_a$ of ghost degree 0. Define $Q$ on the fibers as
\begin{equation}
\begin{gathered}
    Q \phi = \xi^a \phi_a\,, \\ 
    Q \phi_a = - \rho_a{}^b\phi_b
\end{gathered}
\end{equation}
and 
\begin{equation}
    \omega = d({\xi}^{{(3)}}_a \phi^a)d(\phi)\,.
\end{equation}
Since $E$ is trivial as a bundle, there is a lift $\Tilde{\cK}$ of $\cK_B$ determined by the trivialisation and it is easy to see that $i_{\Tilde{\cK}}\omega = 0 $. We also have
\begin{equation}
    i_Q\omega =  \half d({\xi}^{(4)}\phi^k\phi_k) + \alpha\,,\qquad \alpha \in \cI_B\,,
\end{equation}
so that all condition for gPDE over background are satisfied. Introducing coordinates on the space of sections by $\sigma^*\phi = \phi(x)\,,\sigma^*\phi^a = \phi^a(x)$ and for the background gravity fields as $\lorcon^{ab}=\sigma^*(\rho^{ab})$, $e^a=\sigma^*(\xi^{a})$, the AKSZ-like action takes the form:
\begin{equation}
S[e,\lorcon,\phi,\phi_a]=\int_{T[1]X}(e^{(3)}_a \phi^a\dx \phi - \half e^{(4)} \phi^a \phi_a))\,,
\end{equation}
so it is indeed the action of the scalar field coupled to the generic Riemannian background. 

The standard gauge symmetries associated to the background is the local Lorentz symmetry and the diffeomorphisms. In the setup of Section \bref{sec:capturing} the former is generated by  $\commut{Q}{Y}$ with $Y_L=\epsilon^{ab}(x)\dl{\rho^{ab}}$ while the later by $Y=-{\bar\epsilon}^\mu(x)\dl{\theta^\mu}$ and both are compatible with the presymplectic structure. Note that in the later case the parameter is not vertical and hence the second term in the transformation law \eqref{gaugetransf2} gets the addition contribution from $y=\pi_*Y$. For instance, for the gauge transformation of frame field $e^a$ we have:
\begin{multline}
\delta_Y e^a =(\delta_Y\sigma^* )\xi^a=\sigma^*\commut{Q}{Y}\xi^a-\commut{\dx}{y} \sigma^* \xi^a =
\\
\dx(\iota_{\bar\epsilon} e^a) + \iota_{\bar\epsilon} (\dx e^a)=L_{\bar\epsilon} e^a
\,,
\end{multline}
where by some abuse of notations we identified functions on $T[1]X$ with differential forms on $X$ so that
$\iota_{\bar\epsilon} e$ is the contraction of $\bar\epsilon^\mu\dl{x^\mu}$ with $e^a_\mu dx^\mu$.  It is easy to see that the same holds for the remaining fields $\lorcon^{ab},\phi,\phi_a$ so that we get the standard action of the diffeomorphisms on our fields. 

Let us note that it is the standard fact, see e.g.\cite{Didenko:2014dwa}, that the naive gauge transformation of $e^a$, determined in our language by $Y=\epsilon^a(x)\dl{\xi^a}+\epsilon^{ab}(x)\dl{\rho^{ab}}$ still contains diffeomorphisms for $e^a$ if one restricts to configurations with with vanishing torsion, i.e. $\dx e^a+\lorcon^a_b e^b=0$. More precisely, taking $\epsilon^a=\iota_{\bar\epsilon}e^a$ and $\epsilon^{ab}
=\iota_{\bar\epsilon}\lorcon^{ab}$ one finds that $\delta_{\bar\epsilon} e^a=L_{\bar\epsilon}e^a+\iota_{\bar\epsilon} (de^a+\lorcon^a_b e^b)$. Note also that this gauge transformation of $\lorcon^{ab}$ coincides with $L_{\bar\epsilon}\lorcon^{ab}$ only if the Lorentz component of the curvature vanishes.

\subsection{Linearized systems}

Just like in the case of gPDEs discussed in Section~\bref{subsec:Linearized_systems} the linearization of a presymplectic gPDE naturally gives a presymplectic gPDE over background. However, the construction is slightly more involved. 

Consider a presymplectic gPDE $(E,Q,\omega,T[1]X)$.  As a total space of the associated system over background we take $VE \xrightarrow[]{\Pi} E$ - the vertical tangent bundle of $E$. As we discussed in  Section \bref{subsec:Linearized_systems} any projectable vector field on $E$ has a natural lift to $VE$. Applying this to $Q$ gives a vector field $\Tilde{Q}$ on $VE$ such that $\Pi_* \Tilde{Q}=Q$. If $x^b,\theta^a,\psi^C$ are local coordinates on $E$ then 
\begin{equation}
    \Tilde{Q} = Q + q_v\,, \qquad q_v = \phi^A \ddl{Q^B}{\psi^A} \ddl{}{\phi^B}\,.
\end{equation}
where $\phi^A$ are the induced coordinates on the fibres of $VE$.

In addition to the the case of gPDE, now we also need to define a suitable symplectic structure on $VE$, which can be thought of as a linearization of $\omega$. To do this we need a bare affine symmetric connection $\Gamma$ on $E$, which is assumed compatible with the bundle structure, i.e. covariant derivative of a vertical vector field along the vertical direction gives a vertical vector field. In fact the result does not depend on the choice of the connection up to a natural equivalence so that the procedure is natural in this sense. Given such a connection, consider vector field $V = \phi^A \nabla_A$ where $\nabla_A=\dl{\psi^A}-\Gamma_{AB}^C \phi^B\dl{\phi^C}$ 
and $\Gamma_{AB}^C$ are the coefficients of the connection in the vertical sector. It is easy to check that $\nabla_A$ is tangent to $VE \subset TE$ and hence $V$ can be considered a vector field on $VE$.

Let $\chi$ be a presymplectic potential for $\omega$, $\omega=d\chi$. 
We define a lift $\tilde \chi$ of $\chi$
to $VE$ as follows:
\begin{equation}
 \Tilde{\chi} = \half L_V L_V \Pi^* (\chi)\,.
\end{equation}
Let us recall, see Remark~\bref{rem:alt-def-chi}, that if we change $\tilde\chi$ to $\tilde\chi\to \chi+\alpha$, $\alpha \in \cI$ and simultaneously shift $\cL\to \cL-\iota_Q\alpha$ this gives an equivalent formulation of the system. In other words $\chi$ can be considered as an equivalence class modulo addition of $1$-forms from the ideal. We have the following:
\begin{prop}
The equivalence class of $\chi$ does not depend on the choice of the affine connection $\Gamma$.
\end{prop}
Recall that $\Gamma$ is assumed to be symmetric and compatible with the bundle structure of $E\to T[1]X$.  The proof is relegated to Appendix~\ref{App:Linearization}.
\begin{prop}
Let $\tilde\omega=d\tilde\chi$ and $\cK_\omega$ denotes a disribution on $E$ generated by vertical vector fields $k$ such that $i_k\omega\in \cI$. Then 
    $(VE,\Tilde{Q}, \Tilde{\omega}) \xrightarrow{\Pi} (E,Q, \cK_\omega)$ is a presymplectic gPDE over background.
\end{prop}
\begin{proof} 
As explained before, since $\cK_{\omega}$ is vertical there is a natural lift $\Tilde{\cK}$ of $\cK_\omega$. Namely, $\Tilde{\cK}$ is generated by the natural lifts $\tilde k$ of  vector fields $k \in \cK_\omega$. Note that $\commut{V}{\tilde k}$ is vertical on $\tilde E$.  The statement ammounts to checking that the defining relations are satisfied, i.e. 
\begin{equation}
    \begin{gathered}
        i_{\Tilde{Q}}\Tilde{\omega} + d\Tilde{\cL} \in \cI_{\cK}\,, \\
        \half i_{\Tilde{Q}}i_{\Tilde{Q}}\Tilde{\omega} + \Tilde{Q}\Tilde{\cL} = 0 \,,
        \\
        i_{\Tilde{k}}\Tilde{\omega} \in \cI_{\cK}
\end{gathered}
\end{equation}
where $\cI_{\cK} = <\Pi^*\alpha>, \quad \alpha\in \bigwedge^1E: i_k\alpha = 0, \quad \forall k\in \cK_\omega $.

In order to check the first condition we calculate:
\begin{multline}
   \quad i_{\Tilde{Q}}\Tilde{\omega} = 
    \half i_{\Tilde{Q}}L_VL_V \Pi^* \omega = \\ \half(L_V L_V i_{\Tilde{Q}} \Pi^*\omega +(i_{[\Tilde{Q},V]})L_V \Pi^*\omega +  L_V i_{[\Tilde{Q},V]}\Pi^*\omega) \,.\quad 
\end{multline}
The first term equals $d(L_VL_V\Pi^*(H)) + \cI_{\cK}$ because $i_Q \omega = dH + \alpha, \alpha \in \cI$ and $\alpha$ is trivially annihilated by $\cK_\omega$ since $\cK_\omega$ is vertical.  
The last term is zero because $[\Tilde{Q},V]$ is $\Pi$-vertical and the second equals to $i_{[[\Tilde{Q},V],V]}\Pi^*\omega$ and hence belongs  $\cI_K$. 

That $\half i_{\Tilde{Q}}i_{\Tilde{Q}}\Tilde{\omega} + \Tilde{Q}\Tilde{\cL} = 0$ is a matter of direct check. Namely:
\begin{multline}
    i_{\Tilde{Q}}i_{\Tilde{Q}}\Tilde{\omega}  = \half i_{\Tilde{Q}}i_{\Tilde{Q}}L_V L_V \Pi^*\omega = \\= \half i_{\Tilde{Q}}( L_V L_V i_{\Tilde{Q}} \Pi^*\omega + i_{[[\Tilde{Q},V],V]}\Pi^*\omega) = \\ =  \half (L_V i_{\Tilde{Q}} L_V i_{\Tilde{Q}} \Pi^*\omega -i_{[[\Tilde{Q},V]} L_V i_{\Tilde{Q}} \Pi^*\omega + i_{[[\Tilde{Q},V],V]} i_{\Tilde{Q}}\Pi^*\omega ) = \\ = \half(L_V L_V i_{\Tilde{Q}}i_{\Tilde{Q}} \Pi^*\omega  + 2i_{[[\Tilde{Q},V],V]} i_{\Tilde{Q}}\Pi^*\omega )  = \\ = \half (L_V L_V i_{\Tilde{Q}}i_{\Tilde{Q}} \Pi^*\omega  - 2L_{[[\Tilde{Q},V],V]}\Pi^*\cL) \,,
\end{multline}
where we again extensively used the fact that  $i_{[\Tilde{Q},V]}\Pi^*\omega$ is zero because $[\Tilde{Q},V]$ is $\Pi$ vertical. We also have 
\begin{equation}
    \tilde{Q}\tilde{\cL} = \half L_{\Tilde{Q}}L_V L_V \Pi^*\cL = \half (L_V L_V L_{\Tilde{Q}}\Pi^*\cL + L_{[[\Tilde{Q},V],V]}\Pi^*\cL) \,.
\end{equation}

Finally, the last condition  $i_{\Tilde{k}}\Tilde{\omega} \in \cI_{\cK}$ holds because
\begin{equation}
    i_{\Tilde{k}}\Tilde{\omega} = i_{[[\Tilde{k},V],V]}\Pi^*\omega \in \cI_{\cK}\,.
\end{equation}
\end{proof}

\subsubsection{Example: linearized gravity } \label{Sec:lin_grav_presymp}
As an example we employ the above procedure to derive the usual linearized formulation of Palatini-Cartan-Weyl formulation of Einstein gravity. We start with a presymplectic BV-AKSZ formulation of gravity recalled in Example~\bref{example:presymp-grav} and use the notations introduced there. It is useful to identify $VE$  as  $\algg[1]\times \algg[1] \times T[1]X$ so that the fiber of $VE$ is another copy of $\algg[1]$ whose coordinates associated to the same basis are denoted as
\begin{equation*}
    h^a, f^{ab}\,.
\end{equation*}
The vector field $V$ takes the form $V = h^a\ddl{}{\xi^a}+ f^{ab}\ddl{}{\rho^{ab}}$ and the lifted differential $\tilde Q$ acts as:
\begin{equation}
    \tilde Q h^a = -\rho^{ak} h_k - f^{ab}\xi_b \,, \qquad 
    \tilde Q f^{ab} = -f^{ak}\rho_k{ }^b - \rho^{ak}f_k{ }^{b}\,,
\end{equation}
while the linear presymplectic structure on $VE$ is given by:
\begin{equation}
    \tilde\omega = \half L_VL_V\Pi^*\omega = \epsilon_{abcd}(d(h^{a} \xi^b)d f^{cd} + \half d(h^a h^b) d\rho^{cd})\,.
\end{equation}
Although $L_{\tilde Q} \tilde \omega \in \tilde \cI$ holds by construction it can be easily checked directly.
The same applies to the premaster equation $\frac{1}{2}i_Q i_Q \omega + Q\cL = 0$ as well. The AKSZ-like action for the resulting system over background takes the form: 
\begin{equation}
\begin{gathered}
    S[\sigma] = \int \epsilon_{abcd}(\bar{h}^{a} e^b (\dx \bar{f}^{cd} + 2\lorcon^c{ }_k \bar{f}^{kd})  + \half e^a e^b \bar{f}^{c}{ }_k \bar{f}^{kd} + \\ + \half \bar{h}^a \bar{h}^b (d\lorcon^{cd} + \lorcon^c{ }_k \lorcon^{kd}))\,.
\end{gathered}
\end{equation}
where $\lorcon_{ab} = \sigma^*\rho_{ab}\,, e^a = \sigma^* \xi^a$ and $\bar{f}^{ab} = \sigma^* f^{ab}\,, \bar{h}^a = \sigma^* h^a$.
This is easily seen to coincide with the linearization of the Palatini-Cartan-Weyl action, see e.g. \cite{Vasiliev:1980as, Zinoviev:2003ix}, and it is consistent on any Einstein background $e^a, \lorcon^{ab}$.

\section{Gauging global symmetries}

From the perspective of gPDEs over background, the global symmetries associated to the background arise as gauge symmetries preserving a given background solution. In their turn, such gauge symmetries can be understood as a result of gauging of the respective global symmetries of a system in a fixed background. As we saw in Section~\bref{subsec-bgrd-gPDE}, in the non-Lagrangian case any subalgebra of global symmetries can be gauged in a straightforward way, resulting in the associated gPDE over background. In so doing, each symmetry gives rise to the associated background 1-form field. In this Section we study gauging in a more restrictive setup of Lagrangian system by employing presymplectic gPDE approach.

\subsection{Global symmetries and Noether theorem}
Let us first see how global symmetries can be described in the presymplectic gPDE approach. As far as generic symmetries are concerned we limit ourselves to gauge PDEs equipped with the compatible presymplectic structure.

\begin{definition} \label{def:comp-sym}\cite{Grigoriev:2023kkk}
    Given a gPDE $(E,Q,T[1]X)$ equipped with a compatible presymplectic structure $\omega$, vertical $V$ is called a compatible symmetry if

    i) $[Q,V] =0$, 
    
    ii) $\cL_{V}\omega + L_Q d \alpha  \in \cI$
    for some 1-form $\alpha$ on $E$.

\end{definition}
In other words we talk about symmetries of the underlying gPDE that are also compatible with the presymplectic structure. Let us stress that such $V$ does not necessary define a symmetry of the corresponding  AKSZ-like action. At this stage we do not require $\gh{V}=0$ not to exclude the lower-degree symmetries.

Compatible symmetries defined by \bref{def:comp-sym} produce currents which are conserved on the solutions of the underlying gPDE. Indeed, it follows from equation ii) that 
\begin{equation} \label{current-defining}
    i_V\omega + dJ_V - L_Q \alpha + \cI = 0
\end{equation}
for some $J_V\in \cC^{\infty}(E), \quad \gh{J_V} = \gh{V}+n-1$.  $J_V$ is the current associated to the compatible symmetry $V$. Applying $L_Q$ to both sides of~\eqref{current-defining} one finds $d(QJ_V) \in \cI$ and hence $QJ_V=\pi^*(h)$ for some $h \in \cC^\infty(T[1]X)$ satisfying $\dx h=0$.  Given a compatible symmetry $V$, the associated current is far from being unique. Indeed, it is easy to check that equation \eqref{current-defining} enjoys the following symmetry 
\begin{equation}
    J_V\to  J_V + Qg + \pi^*(f)\,, \qquad \alpha \to \alpha - dg
\end{equation}
where $g \in \cC^\infty(E)$ and $f\in  \cC^\infty(T[1]X)$ and hence $J_V$ is defined modulo such an ambiguity. In particular, if $h$ in $QJ_V=\pi^*(h)$ is $\dx$-exact, the ambiguity can be  used to set $QJ_V=0$. This equation implies conservation on-shell. Indeed, if $\sigma_s$ is a solution of $(E,Q)$ then $d_X\sigma_s^*J_V = \sigma_s^*Q J_V = 0$. The additional ambiguity $J_V\to J_V+u$ is described by solutions to $du+L_Q\beta=0$.

Two symmetries related by $V \to V+\commut{Q}{Y}$ with vertical $Y$ are to be considered equivalent. Indeed, the transformed $V$ is still a symmetry and, moreover, its associated current lies in the same equivalence class because
\begin{equation}
\begin{gathered}
    i_{[Q,Y]}\omega = L_Q i_Y \omega - i_YL_Q\omega = L_Q i_Y \omega + \cI\,,
\end{gathered}
\end{equation}
The term $L_Qi_Y\omega$ can be absorbed by $\alpha \to \alpha + i_Y\omega$. 
It follows, \eqref{current-defining} defines a map from equivalence classes of symmetries to equivalence classes of conserved currents. This gives a version of a Noether theorem for gPDEs equipped with compatible presymplectic structures. Note that symmetries and conservation laws are naturally modules over the closed forms on $X$. For instance if $V$ is a symmetry and $\alpha\in \cC^{\infty}(T[1]X)$, $\dx\alpha=0$ then $\pi^*(\alpha)V$ is also a symmetry. Passing to the equivalence classes makes them modules over the de Rham cohomology of $X$.

\subsection{Gauging internal symmetries}\label{sec:internal-symmetries}
Gauging the symmetries originating from space-time transformations generally affects the presymplectic structure and will be discussed separately. Now we concentrate on symmetries that leave the presymplectic structure intact and define them for generic presymplectic gPDEs:
\begin{definition}\label{def:internal-sym}
    Let $(E,Q,\omega,T[1]X)$ be a presymplectic gPDE and $\cK$ be the vertical kernel of $\omega$.  An internal symmetry is a vertical vector field $V$ on $E$ such that
    \begin{equation} \label{comp_symmetry}
        L_V\omega \in  \cI \,, \qquad 
            [V,Q] \in \cK\,.
    \end{equation}
\end{definition}
The infinitesimal symmetry transformation of a section $\sigma$, determined by $V$ is defined as $\delta \sigma^* = \sigma^* \circ L_{V}$. Recall that we do not require $V$ to have vanishing ghost number. However, $V$ of a nonvanishing ghost degree does not determine a nontrivial transformation of sections.  

In the setup of presymplectic gPDEs one can also introduce a concept of internal conserved currents. More precisely, function $J$ on $E$ is called an internal conserved current if $QJ=0$ and $\cK J=0$. Note that we do not restrict the ghost degree to be $n-1$ and hence allow for currents of generic degree. Unless it leads to confusions in what follows we often omit the adjective ``internal''. If $\sigma$ is a solution to the Euler-Lagrange equations of the AKSZ-like action for $(E,Q,\omega,T[1]X)$ then $\dx\sigma^*(J)=0$, i.e. $J$ defines a conserved current of the underlying Lagrangian system. Indeed, EL equations can be written as $d_X\sigma^* = \sigma^*(Q+k)$ for some $k\in \cK$ and hence $d_X\sigma_s^*J = 0$ thanks to $\cK J=0$. Moreover, if $J=QN$ and $\cK N=0$ then its associated charge is trivial in the sense that $\sigma^*(J)=\dx(\sigma^*N)$ and hence $Q$-exact currents are to be considered trivial.

In the limited setup of presymplectic gPDEs and their internal symmetries and currents, one can define a version of the Noether map. To simplify the discussion let us restrict ourselves to the case where the de Rham cohomology of $X$ and the 1st cohomology of $E$ are trivial. These conditions can be always meet by resorting to a local analysis. If $V$ is an internal symmetry then equation~\eqref{comp_symmetry} implies 
\begin{equation} \label{conseq-cur}
    i_V\omega + dJ_V + \cI = 0\,,
\end{equation}
for some $J_V \in \cC^\infty(E)$. It turns out that $J_V$ can be chosen in such a way that it is conserved on solutions of the Lagrangian system determined by $(E,Q,\omega,T[1]X)$. Indeed, applying $L_Q$ to both sides of~\eqref{conseq-cur} one finds $d(L_QJ_V)\in\cI$ so that by adding a term of the form $\pi^*f$ one can achieve $QJ_V=0$. Moreover, applying $k\in\cK$ to both sides of~\eqref{conseq-cur} and using $\cK\pi^*f=0$ one finds that
$\cK J_V=0$.

It is natural to consider two internal symmetries equivalent, if their difference is of the form $[Q,Y]+Z$ with $L_Y\omega \in \cI$ and $Z\in \cK$. While $Z$ clearly does not contribute to the equation \eqref{conseq-cur} the first term does. However, adding this term can be compensated by  $J \xrightarrow{} J + QN_Y$. Indeed, $L_Y\omega \in \cI$ implies that $i_Y\omega + dN_Y \in \cI$ which in turn defines $N_Y$ up to addition of functions of the form $\pi^*f$, resulting in $i_{[Q,Y]}\omega +dQN_Y\in \cI$. 

As we discussed above, there is also a natural equivalence relation on conserved currents such that two currents are considered equivalent if they differ by a $Q$-exact term. The above discussion implies that the map from internal symmetries to conserved currents, determined by \eqref{conseq-cur} is well-defined on equivalence classes. This can be regarded as a somewhat restricted version of the usual Noether map. Indeed, we have seen that internal conserved currents are conserved on solution to the underlying Lagrangian system. Likewise, the internal symmetries define symmetries of the corresponding action:
\begin{prop}
    An internal symmetry defines a global symmetry of the presymplectic action. 
\end{prop}
\begin{proof}
    The statement is nontrivial only in the case $\gh{V}=0$.   We have
\begin{equation}
    L_{V} \chi = i_{V} d \chi + d(i_{V}\chi) = d(i_{V}\chi - J_V) - \beta 
\end{equation}
for some $J_V\in \cC^\infty(E)$ and $\beta \in \cI$. Recall that under our assumptions we can assume $QJ_V=0$. Furthermore, 
\begin{equation}
\begin{gathered}
    L_{V} \cL = i_{V}d\cL = -i_{V} i_Q \omega - i_{V}(\cI) = i_Q(dJ_V + \beta ) = QJ + i_Q\beta\,.
\end{gathered}
\end{equation}
The variation of the integrand $(\sigma^*(\chi))(\dx)+\sigma^*(\cL)$ takes the form
\begin{equation} 
\begin{gathered}
    \sigma^*(d(i_{V}\chi - J_V))(\dx)  - \sigma^*(\beta)(\dx) + \sigma^*(i_Q\beta) \,.
\end{gathered}
\end{equation} 
The first term is a total derivative while the other two terms cancel each other. 
\end{proof}
It is important to stress that internal conserved currents do not generally exhaust all the conserved currents of the underlying Lagrangian system. The same also applies to internal symmetries. Nevertheless, as we have just seen they are related by the Noether map. Let us also note that the above discussion can be generalised to the case where de Rham cohomology of $X$ is not empty. In this case, however, an internal symmetry generally defines a quantity which is conserved only locally on $X$.

We now turn to the gauging procedure:
\begin{theorem} \label{thm:gauging}
    Let $(E,Q,\omega,T[1]X)$ be a presymplectic gPDE  and $V_\alpha$ be its internal symmetries. Let in addition  
    $V_\alpha$ form a representation of a (graded) Lie algebra $\mathfrak{g}$. Then  $\Tilde{E} = E \times \mathfrak{g}[1]$ is naturally a background presymplectic gPDE (possibly anomalous) which we call the gauging of the global symmetries $V_\alpha$.
\end{theorem}
As before, we assume that the de Rham cohomology of $X$ vanishes.
\begin{proof}
Let $e_\alpha$ denote a basis in $\algg$ such that $V_\alpha$ represents $e_\alpha$ and $c^\alpha$, $gh(c^\alpha) = 1-\gh{V_\alpha}$ be the associated coordinates on $\mathfrak{g}[1]$. As a $Q$-structure on $\Tilde{E}$
we take:
\begin{equation}
    \Tilde{Q} = Q + \dg\,.
\end{equation}
Of course $\Tilde{E}$ is a bundle over $B = T[1]X \times \mathfrak{g}[1]$ and it is easy to see that  $\Tilde{Q}$ projects to $\gamma = \dx + \dg^0$ and $\gamma^2 = 0$ due to the Jacobi identity. The symplectic structure on $\Tilde{E}$ is taken to be the pullback of $\omega$ from $E$ and we keep denoting it by $\omega$.

Now we need to check that there exists $\tilde\cL$ such that $\tilde Q,\omega,\tilde \cL$ satisfy the defining relations. Restricting for simplicity to the case $\gh{V_\alpha}=0$ and taking $\tilde\cL =\cL-c^\alpha J_\alpha$, where $J_\alpha$ are conserved currents associated to $V_\alpha$ one finds:
\begin{equation} \label{gauging_hamilton_cond}
    i_{\Tilde{Q}}\omega + d(\cL) - d(c^\alpha J_\alpha) + d(c^\alpha)J_\alpha + \cI[dx,d\theta] = 0
\end{equation}
Because $d(c^\alpha)$ lies in $\cI_B$ we conclude that $i_{\Tilde{Q}}\omega + d\cL\in \cI_B$.  Of course, equation \eqref{gauging_hamilton_cond} defines $\Tilde{\cL}$ only up to  functions which are pulled back from $B$.

The remaining condition is the presymplectic master equation. We have
\begin{equation}
    \frac{1}{2}i_{\Tilde{Q}} i_{\Tilde{Q}} \omega = \frac{1}{2} i_{Q + cV} i_{Q+cV} \omega = \frac{1}{2}i_Q i_Q \omega + i_Q i_{cV}\omega + \frac{1}{2}i_{cV}i_{cV} \omega  \,.
\end{equation}
Furthermore, $i_{cV}i_Q\omega = -c^\alpha V_\alpha \cL$ and $\frac{1}{2}i_{cV}i_{cV}\omega = \frac{1}{2}c^\beta c^\alpha V_{\beta}(J_\alpha)$. 
Next we calculate:
\begin{equation}
    \Tilde{Q}\Tilde{\cL} = (Q + q)(\cL - c^\alpha J_\alpha) = Q\cL + c^\alpha V_\alpha \cL + c^\alpha Q J_\alpha + \frac{1}{2}f^\alpha{ }_{\beta\gamma}c^\beta c^\gamma J_\alpha - c^\beta c^\alpha V_\beta J_\alpha 
\end{equation}
It can be shown that (see Appendix~\ref{App:closure}) 
\begin{equation}
    V_\alpha J_\beta = f^\gamma{ }_{\alpha\beta}J_\gamma + \pi^*(f_{\alpha\beta}(x,\theta))
\end{equation}
for some $f_{\alpha\beta} \in C^\infty(T[1]X)$
So all in all:
\begin{equation}\label{backg-premapster-gen}
    \frac{1}{2}i_{\Tilde{Q}} i_{\Tilde{Q}} \omega + \Tilde{Q}\Tilde{\cL} + \frac{1}{2}c^\alpha c^\beta \pi^*(f_{\alpha\beta}) = 0\,.
\end{equation}
In other words, the presymplectic master equation is satisfied up to an extra term which is a function pulled-back from $B$. In  Appendix~\bref{App:closure} we show that it is $\gamma$-closed. If the extra term is $\gamma$-exact, then it can be eliminated by adding the respective term to $\Tilde{\cL}$. Otherwise it is a cohomology, related to the classical anomaly. In any case this term is equal to zero on the space of background solutions due to degree reasoning, so the system always defines a presymplectic gPDE  when pulled-back to a given background solution. What anomaly means is that some of the  background symmetries become anomalous after gauging. The anomaly can be eliminated by replacing $\algg$ with its suitable central extension.
\end{proof}

\subsubsection{Example: gauging scalar multiplet}\label{sec:scalar-mult}
Consider a multiplet of free scalar fields. In the formalism of presymplectic gPDEs, it can be described by $E\rightarrow T[1]X$ where $E$ is a trivial fibre bundle with fibre coordinates $\phi, \phi_a$ taking values in a particular module $W$ of a real Lie algebra $\mathfrak{g}$, equipped with an invariant positive-definite inner product. We work in the orthonormal basis $e_I$, $\inner{e_I}{e_J}=\delta_{IJ}$ and denote the respective component fields by $\phi^I,\phi^I_a$. The $Q$-structure is given by  
\begin{equation}
    Q \phi^I = \theta^a \phi^I_a\, \qquad Q\phi^I_a = 0
\end{equation}
and the presymplectic structure is 
\begin{equation}
    \omega = d\chi\,, \qquad \chi= \theta^{(3)}_a \inner{\phi^{a}}{d\phi}\,,
\end{equation}
The associated action is just the sum of $n$ copies of the scalar field action.

Let $\rho(t_\alpha)$ be a linear operator on $W$ representing the basis element $t_\alpha\in\algg$. The associated internal symmetry is given by the following vertical vector field
\begin{equation}
V_\alpha \phi=\rho(t_\alpha) \phi\,, \qquad V_\alpha \phi^a=\rho(t_\alpha) \phi^a\,,
\end{equation}
or in terms of components: $V_\alpha =  t_{\alpha, IJ}(\phi^I\ddl{}{\phi^J} + \phi^I_a\ddl{}{\phi^J_a})$. One has:
\begin{equation}
    i_{V_\alpha} \omega = - d\left (\theta^{(3)}_a \inner{\rho(t_\alpha)\phi^a}{\phi}\right) + \cI = -dJ_\alpha + \cI\,.
\end{equation}

Following the general prescription given in Theorem~\bref{thm:gauging} we take  $\Tilde{E} = E \times \algg[1]$ and 
\begin{equation}
    \Tilde{Q} = Q + c^\alpha\rho(t_\alpha) - \frac{1}{2}f^\alpha_{\beta \gamma}c^\beta c^\gamma \ddl{}{c^{\alpha}}\,, \qquad \tilde \cL=\cL - c^\alpha J_\alpha\,,
\end{equation}
where $c^\alpha$ are coordinates on $\algg[1]$ corresponding to basis $t_\alpha$. It is easy to check that 
the defining conditions are satisfied: $i_{\Tilde{Q}}\omega + d\tilde\cL\in \cI_B$ and $\frac{1}{2}i_{\Tilde{Q}}i_{\Tilde{Q}}\omega + \Tilde{Q}(H - c^\alpha J_\alpha) = 0$ so that we have a well-defined background system. Note that in this simple example the anomaly is not present.

Strictly speaking solutions of the background system are flat $\algg$-connections. However, this can be easily generalized by picking a suitable distribution $\cK_B$ on $T[1]X\times \algg[1]$. For instance to conider generic connections one simply takes $\cK_B$ generated by $\theta^a\theta^b\dl{c^\alpha}$. It is easy to check that it precisely removes the zero-curvature equation without affecting gauge tranformations. Moreover, this distribution doesn't affect $\phi,\phi^a$ sector and hence the result is again a consistent presymplectic gPDE over background. Restricting to a background solutions, i.e. a generic $\mathfrak{g}$ connections parameterized by $S^*(c^\alpha) = A^\alpha_\mu \theta^\mu$, the action takes the following form:
\begin{equation}\label{gauged-scalar}
    S = \int d^4x \,\,\inner{\phi^a}{\partial_a\phi -\half \phi_{a} - A^{\alpha}\rho(t_{\alpha})\phi_a}\,.
\end{equation}
This is of course just a first order form of the usual action of the scalar multiplet in the background of a Yang-Mills field.


\subsection{Homogeneous presymplectic gPDEs over background}

Although we refrain from discussing general space-time symmetries of presymplectic gPDEs and their gauging, there is a very special but still quite rich class of presymplectic gPDEs over background that can be seen as a result of gauging space-time symmetries. This is the Lagrangian counterpart of homogenous gauge PDEs over background described in Section~\bref{sec:homgPDE}.

Let the spacetime manifold be a homogeneous space $X = G/H$ and $\algg$ and $\algh$ denote the respective Lie algebras. Let  $(F,q)$ be an almost $Q$-manifold equipped with the linear map $\rho:\algg\to \mathrm{Vect}(F)$. If $t_\alpha$ is a basis in $\algg$ we define "fundamental" vector fields $V_\alpha$ on $F$ by $V_\alpha=\rho(t_\alpha)$. Let us stress that for the moment we do not require $\rho$ to be a homomorphism. In fact it will be required to satisfy the homomorphism condition up to a distribution.

Now consider the following almost-$Q$ bundle $\tilde F=F\times \algg[1] \to \algg[1]$ where $(\algg[1],\dg)$ is a $Q$-manifold associated to $\algg$ and the total almost $Q$-structure given by:
\begin{equation}
 \Tilde{q} = q + \dg+c^\alpha V_\alpha \,, \qquad \dg=-\frac{1}{2}f^\alpha_{\beta \gamma}c^\beta c^\gamma \ddl{}{c^\alpha}\,,
\end{equation}
where we introduced linear coordinates $c^\alpha$ on $\algg[1]$. In addition, suppose that $\tilde F$ is equipped with the presymplectic structure $\omega$, $\gh{\omega} = \dim X - 1$ such that 
\begin{equation}
\label{pgpde}
        d\omega = 0, \qquad i_Q\omega + d\cL \in \cI_B, \qquad \frac{1}{2}i_Q i_Q \omega + Q\cL = 0\,,
\end{equation}
where $\cI_B$ is an ideal generated by forms $\pi^*\alpha$. In other words $(\tilde F,\tilde q, \omega)$ satisfies all the axioms of a presymplectic gPDE, but with $(T[1]X,d_X)$ replaced with $(\algg[1], \dg)$. As in the case of gPDEs over $T[1]X$  it is easy to check that $Q^2 \in \cK_{\omega}$.  Note that conditions \eqref{pgpde} do not generally imply that $V_i$ form a representation of $\mathfrak{g}$ but only a weaker property $c^\alpha c^\beta ([V_\alpha, V_\beta] - f^\gamma{ }_{\alpha \beta} V_\gamma) \in \cK_\omega$.

Finally, $(E,Q) = (\tilde F,\tilde q) \times (T[1]X,\dx)$ is a presymplectic gPDE over background if one takes  $Q = \dx + \Tilde{q}$ and a presymplectic structure being the pullback of $\omega$ to the total space. It is natural to call such systems a homogeneous presymplectic gPDE over background because it can be thought of as a presymplectic version of homogeneous gPDEs considered in Section~\bref{sec:homgPDE}. 

A canonical flat Cartan connection on $G/H$ gives a solution to the background gPDE and the pullback of $E$ to this solution defines a presymplectic gPDE over $T[1]X$. It follows, the gauge transformations of the form $W(\epsilon)=\commut{Q}{\epsilon^\alpha(x)\dl{c^\alpha}}$, where $\epsilon=\epsilon^\alpha e_\alpha$
is covariantly constant $\algg$-valued function, define the global $\algg$-symmetries of the AKSZ-like action of the theory, see Section~\bref{sec:capturing}. Note that although the initial map $\rho:\algg \to \mathrm{Vect}(F)$
was not necessarily a representation, the resulting global symmetries $W(\alpha)$ necessarily form $\algg$:
\begin{equation}
\commut{W(\alpha)}{W(\beta)}=\commut{Q}{\gamma^\mu\dl{c^\mu}}\,, \qquad 
\gamma=\commut{\alpha^\mu\dl{c^\mu}}{\commut{Q}{\beta^\nu\dl{c^\nu}}}=\commut{\alpha}{\beta}_\algg^\mu\dl{c^\mu}\,.
\end{equation}
where $\commut{\alpha}{\beta}_\algg$ is the Lie bracket in $\algg$. The last equality is easily checked by observing that the only term in $Q$ quadratic in $c$ is $\dg$. Notice that the second equality defines  a bracket on the space of parameters and this bracket is a particular example of so-called derived brackets. Related structures were discussed in~\cite{Grigoriev:1998gn,Grigoriev:2000zg,Barnich:2001jy}.

\subsection{Example: conformally-coupled scalar} \label{sec:conf-scal}

To give an example of a homogeneous presymplectic gPDE let us start with a  gPDE description of a conformal scalar
field on Minkowski space. The respective gPDE is $E=F\times T[1]X$, where $F$ is a the linear space of solutions to 
\begin{equation}
    \ddl{}{y^a}\ddl{}{y_a}\phi = 0
\end{equation}
in the space of formal power series in commuting variables $y^a$. The differential is given by the usual total derivative whose action on the fibre coordinates encoded in $\phi(y)$ can be written as $D_a \phi(y)=\dl{y^a}\phi(y)$ so that:
\begin{equation}
\label{cscalar}
Q=\theta^a D_a=\dx+ \theta^a(\dl{y^a}\phi(y))\dl{\phi(y)}\,.
\end{equation}
This gives a minimal gPDE decription of the conformal scalar. Because there is no gauge invariance in this case it coincides with the description of this PDE as a bundle equipped with the Cartan distribution, see e.g.~\cite{vinogradov:2001,Krasil'shchik:2010ij}, and the unfolded description~\cite{Vasiliev:1980as,Lopatin:1987hz} of this system.

To obtain a Lagrangian description of the system one equips $E$ with the compatible presymplectic structure $\omega=d(\theta^{(3)}_a \phi^a d\phi)$, where $\phi,\phi_a$ are introduced by $\phi(y)=\phi+y^a\phi_a+\ldots$. It is easy to see that it makes $E$ into a presymplectic gPDE and the respective action is just \eqref{gauged-scalar} for just one scalar and $A=0$. Of course, only $\phi$ and $\phi_a$ coordinates enter the action and the presymplectic structure so that one can disregard all the higher coordinates on $F$ resulting in the minimal presymplectic gPDE formulation of the scalar discussed in Section \bref{sec:scalar-mult}. However, for the moment we keep all the coordinates in order to maintain the conformal symmetry in a manifest way.

Because we are dealing with a conformal-invariant system the fiber $F$ is a module over the conformal algebra $\algg$. More precisely the 
action of the generators $t_\alpha$ of the conformal algebra on $F$ can be defined using differential operators in $y$-space as follows: 
\begin{equation}
    \begin{gathered}
        P_a \phi(y) = \ddl{}{y^a}\phi(y)\, \qquad K_a \phi(y) = (2y_a(y^c\ddl{}{y^c} + 1) - y^c y_c \ddl{}{y^a})\phi(y) \,\\
         \\ J_{ab}\phi(y) = y_{[b}\ddl{}{y^{a]}}\phi(y) \,, \qquad D\phi(y) = -(y^c\ddl{}{y^c} + 1)\phi(y)\,.
    \end{gathered}
\end{equation}

Because $P_a$ coincides with total derivatives on $F$ it is easy to rewrite the system as a homogeneous gPDE over background. Namely, taking $E=F\times \algg[1]\times T[1]X$ and $Q=\dx+\dg$, where $\dg =\dg^0+\xi^aP_a + \half \rho^{ab} J_{ab}+ \half \kappa^a K_a + \lambda D$ is the CE differential with coefficients in functions on $F$.  Taking a flat Cartan connection as a solution to $(B,\gamma)=(\algg[1]\times T[1]X,\dx+\dg^0)$ and pulling back $E$ to this solution one recovers the initial gPDE. More precisely, taking the vacuum connection to be $\theta^aP_a$ one precisely gets \eqref{cscalar}.

A simple observation is that $E$ is equipped with a symplectic structure which reduces to the above $\omega$ when $E$ is pulled back to the background solution $\theta^aP_a$. Indeed, $\tilde \omega = d(\xi^{(3)}_a \phi^a d\phi)$ obviously does the job. Moreover, $\tilde\omega$ happen to be $Q=\dx+\dg$-invariant modulo the ideal generated by $d\xi,d\rho,d\kappa,d\lambda$ and one finds:
\begin{equation}
    i_Q\tilde\omega +d\tilde\cL + \cI_B = 0\,, \qquad \tilde\cL= (\xi^{(3)}_a \phi^a \lambda \phi +\frac{1}{2}\xi^{(3)}_a \kappa^a \phi^2 - \frac{1}{2}\xi^{(4)}\phi^k\phi_k ) \,.
\end{equation}
Straightforward but tedious computations show that $\half i_Qi_Q\tilde \omega+Q\tilde \cL=0$ and hence we are dealing with the homogeneous presymplectic gPDE over background. 

By inspecting $\tilde\cL,\tilde\omega$ one observes that among all $\phi_{\ldots}$ only coordinates $\phi,\phi_a$ are involved there. Moreover, if one sets to zero $Q\phi_{ab\ldots}$ for all the remaining coordinates on $F$, the axioms of the presymplectic gPDE over background remain intact. Of course this truncation can't affect the resulting action. More explicitly, the truncation of $Q$ acts as:
\begin{equation}
    \begin{gathered}
        Q\phi = \xi^a\phi_a - \lambda \phi\,, \\
        Q\phi_a = \phi_b \rho^b{ }_a - 2\lambda \phi_a + \kappa_a \phi\,.
    \end{gathered}
\end{equation}

What is less trivial, is that one can safely replace the background gPDE $(T[1]X\times \algg[1],\dx+\dg^0)$ with its weak version obtained by introducing the nontrivial distribution $\cK_B$ generated by:
\begin{equation}
W_{cd}^{ab}(x)\xi^c\xi^d\dl{\rho^{ab}}\,, \qquad C^a_{bc}(x)\xi^b\xi^c\dl{\kappa^{a}}
\end{equation}
where $W_{ab}^{cd}(x)$ and $C^a_{bc}(x)$ are arbitrary functions having the tensor structure of the Weyl and the Cotton tensor respectively. Let us parametrize the space of background sections by $S: T[1]X \rightarrow B$ by 
\begin{equation}
    \begin{gathered}
        S^*(\xi) = e^a{ }_{\mu}(x)\theta^\mu\,, \qquad S^*(\rho^a{ }_b) = \lorcon^a{ }_{b,\mu}(x)\theta^\mu\,, \\
        S^*(\kappa^a) = f^a{ }_{\mu}(x)\theta^\mu \,, \qquad
        S^*(\lambda) = \lambda_{\mu}(x)\theta^\mu\,, \\
    \end{gathered}
\end{equation}
and assume $e^a{ }_{\mu}$ to be invertible. Taking into account the distribution $\cK_B$, the condition that $S$ is a solution implies that all the components of the curvature of $S$ vanishes save for the components associated to Lorentz subalgebra and special conformal transformations which are allowed to be generic Weyl and Cotton tensors respectively. This means that generic solution is a normal Cartan connection describing the conformal geometry. Because the lift of $\cK_B$ belongs to the kernel of the symplectic structure the modified system is still a presymplectic gPDE over background.

Let us write the AKSZ like action of our system:
\begin{multline} \label{conf_scal_aksz_act}
    S = \int d^4x d^4\theta (\sigma^*(\chi)(\dx) + \sigma^*(\cL)) = 
    \\
    = \int_{T[1]X}  ( e^{(3)}_a\phi^a \dx\phi - e^{(4)} \half \phi^k\phi_k + 
     e^{3}_a \lambda^{a} \phi + \frac{1}{2} e^{(3)}_a f^a \phi^2\big)\,.
\end{multline}
For this action to be invariant under background gauge transformations we have to further restrict to only those sections which solve the background EOMs. 

Consider the equations in the sector of Lorentz connection, $\rho$:
\begin{equation} \label{rhoEq}
    \partial_{[\mu}\lorcon^a{ }_{b,\nu]} + \lorcon^a{ }_{k,[\mu}\lorcon^k{ }_{b,\nu]} - e^a{ }_{[\mu}f_{b,\nu]} + e_{b,[\mu}f^a{ }_{\nu]} + \frac{1}{2}W^a{ }_{b\mu\nu} = 0\,.
\end{equation}
Let us denote by $R^a{ }_{b\mu\nu}=\partial_{[\mu]}\lorcon^a{ }_{b,\nu]} - \lorcon^a{ }_{k,[\mu}\rho^k{ }_{b,\nu]}$  the Riemannian curvature expressed in terms of the Lorentz connection. Taking the trace of \eqref{rhoEq} with $e_a{ }^{\mu}$ one finds 
\begin{equation}
    R_{b\nu} = 2f_{b,\nu} + e_{b,\nu}f\,,
\end{equation}
so that $f = \frac{1}{6}R$. For $\lambda$ we have
\begin{equation}
    \partial_{[\mu}\lambda_{\nu]} = e^a{ }_{[\mu}f_{b,\nu]} = 0\,,
\end{equation}
so every $\lambda_\nu = \partial_\nu \alpha$ is a solution. However, as we will see shortly $\lambda$ is a pure gauge field.

Let us examine the background symmetries. They are defined by degree $-1$ vector fields 
\begin{equation}
\begin{gathered}
    Y_{W} = w(x)\frac{\partial}{\partial \lambda} \,, \qquad Y_{k} = k^a(x)\frac{\partial}{\partial \kappa^a} \,, \\ Y_{L} = \epsilon^{ab}(x)\frac{\partial}{\partial \rho^{ab}}\, \qquad Y_\epsilon = \epsilon^\mu(x) \frac{\partial}{\partial \theta^\mu}\,.
\end{gathered}
\end{equation}
They all respect the presymplectic structure and therefore they define some background gauge symmetries. Note that using the $Y_k$ symmetry one can always set to zero the component $\lambda_\mu (x)\equiv \sigma^*(\lambda)$ because we can shift it by $\delta \lambda_{\mu} = k^a(x) e_{a\mu}$. The remaining background symmetries are the local Lorentz rotations, the Weyl rescalings and diffeomorphisms, as it should be. Substituting $f, \lambda$ back into the action \eqref{conf_scal_aksz_act} we get the usual action of the conformally-coupled scalar.

\subsection{Higher-form symmetries in Maxwell theory}
It turns out that our formalism is also suitable for gauging higher-form symmetries. Leaving the development of the general formalism for a future work we now concentrate on the example of 1-form symmetry of Maxwell theory and its gauging. First, we briefly recall how this symmetry emerges. More detailed exposition can be found in~\cite{Hull:2024uwz}, see also~\cite{Gaiotto:2014kfa}.

Let us consider Maxwell theory, dynamical field is the $u(1)$-connection $A$ and its curvature is $F = dA$. The action functional is 
\begin{equation}
    S[A] = \int_X F \wedge * F\,.
\end{equation}
This action possesses a global symmetry $\delta A = \lambda$ where $\lambda$, $d\lambda = 0$ is a closed 1-form defined on the spacetime . If the de Rham cohomology of the spacetime is trivial, then $\lambda=d\nu$ and hence this symmetry is trivial because it coincides with  a gauge transformation. If there is a non-trivial cohomology in degree $1$, then picking $\lambda$ to be a representative of a nontrivial cohomology class produces a genuine global symmetry. 

Gauging this symmetry amounts to promoting the ``gauge parameter'' $\lambda$ to a general 1-form $\beta$ and introducing a 2-form background field $B$. The gauge transformations read:
\begin{equation}
\label{AB-gauge}
    \delta_\beta A = \beta\,, \qquad \delta_\beta B = d\beta\,.
\end{equation}
It is now easy to see that global symmetry with closed $\lambda$ arises as a gauge transformation which preserves a given $B$. 

To introduce the coupling of $A$ to $B$ at Lagrangian level one notes that $dA - B$ is a gauge invariant combination so that the gauge-invariant action can be taken as: 
\begin{equation} \label{gauged_elec}
    S[A,B] = \int_X (dA-B)\wedge * (dA-B)\,.
\end{equation}

Now, let us demonstrate how \eqref{gauged_elec} can be systematically obtained  in the presymplectic gPDE framework by gauging the degree $-1$ symmetry of the Maxwell theory. The presymplectic gPDE formulation of the system is given by $(E,Q,\omega,T[1]X)$ where $X$ is the spacetime with coordinates $x^a$ and the fibre coordinates are:
\begin{equation}
\begin{gathered}
    C,\quad \gh{C} = 1, \qquad F_{ab}, \quad \gh{F_{ab}} = 0\,.
    \end{gathered}
\end{equation}
The $Q$-structure is given by
\begin{equation}
    QC = \half F_{ab}\theta^a \theta^b\,, \qquad QF_{ab} = 0\,, \qquad Qx^a=\theta^a\,.
\end{equation}
The presymplectic structure reads as
\begin{equation}
    \omega = \half dC d(F_{ab}\theta_c \theta_d) \epsilon^{abcd}=d(\theta^{(2)}_{ab} F^{ab})dC\,.
\end{equation}
and the AKSZ-like action is: 
\begin{equation}
    S[A,F] = \int_X F^{ab}(\partial_a A_b - \partial_b A_a) - \half F^{ab}F_{ab}\,.
\end{equation}
The generalization to Yang-Mills case and the explicit form of the resulting BV formulation can be found in~\cite{Grigoriev:2022zlq}.

It is well known that Maxwell theory possesses a global reducibility identity whose associated surface charge is the familiar electric charge, see e.g.~\cite{Barnich:2001jy}. In our approach this is captured by a degree $-1$ symmetry represented by a vector filed $V = \ddl{}{C}$. It is easy to check that
\begin{equation}
    [V,Q] = 0\,, \qquad 
    i_V \omega = \half d(F_{ab}\theta_c \theta_d)  \epsilon^{abcd} = d J_e\,.
\end{equation}
so that $V$ is indeed an internal symmetry and its associated current $J_e$ has degree $2$. The respective conserved charge is given by
\begin{equation}
    Q_e = \int_{T[1]\Sigma_2} \sigma^*(J_e) 
\end{equation}
and is of course proportional to the usual electric charge. Given a de Rham cohomology class $\alpha \in \cC^\infty(T[1]X)$ in degree $1$, one can generate a usual (degree $0$) symmetry $\alpha V$, $\gh{\alpha V} = 0$ and $[\alpha V, Q] = 0$. It follows
\begin{equation}
    i_{\alpha V} \omega = \half d(\alpha F_{ab}\theta_c \theta_d)\epsilon^{abcd} + \cI
\end{equation}
which reproduces the Noether charge of the corresponding global symmetry. 

Let us now apply the gauging procedure explained in Section~\bref{sec:internal-symmetries} to the degree $-1$ symmetry $V$. This is done by extending $T[1]X$ to the background gPDE $(B,\gamma,T[1]X)$ whose fibre coordinate is $b, \gh{b} = 2$, $\gamma b = 0$. To make $B$ dynamical we also equip $B$ with the distribution $\cK_B$ generated by ${\theta^{(3)}_a}\ddl{}{b}$. The total $Q$ structure on $E$ is then given by $\Tilde{Q} = Q + bV$. The presymplectic structure $\omega$ remains invariant, giving 
\begin{equation}
    2i_{\Tilde{Q}}\omega = d( \theta^4 F^{ab}F_{ab}) +  d( bF_{ab}\theta_c\theta_d\epsilon^{abcd}) + \cI\,,
\end{equation}
and the equation of motion of the background system do not constrain $\sigma^*(b)$  thanks to $\cK_B$. We introduce the following notation for the coordinates on the space of sections by $\sigma^*C = A_a(x)\theta^a$, $\sigma^*(b) = B_{ab}(x)\theta^a \theta^b $. The action of the background gauge transformations are generated by $Y = \alpha\ddl{}{b}$ where $\alpha$ is an arbitrary one form. Their action on fields reads
\begin{equation}
\begin{gathered}
    \delta_Y A_a \theta^a = \delta_Y\sigma^*(C) = \sigma^*[Q,Y] C = \sigma^*\alpha \,,\\
    \delta_Y B_{ab}\theta^a \theta^b = \delta_Y \sigma^*b = \sigma^*[Q,Y] b = \sigma^* \dx\alpha\,,
\end{gathered}
\end{equation}
reproducing~\eqref{AB-gauge}. The AKSZ like action of the system takes the form:
\begin{equation}
    S = \int_X F^{ab}(\partial_a A_b - \partial_b A_a    -  B_{ab})  - \half F^{ab}F_{ab}\,,
\end{equation}
which is just the first order formulation of \eqref{gauged_elec}.

\section*{Acknowledments}
We are grateful to T. Basil, N. Boulanger, A. Cattaneo, B. Kruglikov, A. Kotov, A. Mamekin, M. Markov, P. Bielavsky, R. Szabo, A. Verbovetsky, and M.~Zabzine for useful exchanges. We also wish to thanks V. Gritzaenko for the collaboration at the early stage of this project.
The authors wish to thank the hospitality of the Mittag-Lefler Institute during the workshop "Cohomological Aspects of Quantum Field Theory".

\appendix

\section{Transformation of the action} \label{AppA}

First, let us define vector fields and inner products along maps (sections).

\begin{definition}
    Let $\phi: M_1 \rightarrow M_2$ be a smooth map. A vector field $V_\phi$ along $\phi$ is a derivation of functions on $M_2$ with values in $M_1$, s.t.

    i) $V_\phi (fg) =V_\phi (f)\phi^*(g) + (-1)^{|V||f|}\phi^*(f) V_\phi(g)$
\end{definition}

\begin{definition}
    Given $V_\phi$ along $\phi$ we define $i_{V_\phi}$ along $\phi$ to be a map 
    \begin{equation*}
        i_{V_\phi}: \bigwedge^{k}(M_2) \rightarrow \bigwedge^{k-1}(M_1) 
    \end{equation*}
    obeying

    i) $i_{V_\phi}f = 0, \quad \forall f \in C^\infty(M_2)$ 
    
    ii) $i_{V_\phi}d(f) = V_{\phi}f, \quad \forall f \in C^\infty(M_2)$ 

    iii) $i_{V_\phi} \alpha \wedge \beta = (i_{V_\phi} \alpha) \wedge \phi^*\beta + (-1)^{|V+1||\alpha|}\phi^*\alpha \wedge i_{V_\phi} \beta $
\end{definition}
We also define 
\begin{definition}
    Given $V_\phi$ along $\phi$ with $\gh{V_\phi} = 1 \quad \text{mod}(2)$ an operation $i^2_{V_\phi}$ 
    \begin{equation*}
        i^2_{V_\phi}: \bigwedge^{k}(M_2) \rightarrow \bigwedge^{k-2}(M_1) 
    \end{equation*}
is obeying: $\forall f \in C^\infty(M_2), \quad \forall \alpha, \beta \in \bigwedge (M_2)$ 

i) $i^2_{V_\phi} f = 0$

ii) $i^2_{V_\phi} df = 0$

iii)\begin{equation*}
    \begin{gathered}
        i^2_{V_\phi} (\alpha \wedge \beta) = i^2_{V_\phi} (\alpha) \wedge \phi^*(\beta)  + 2 (i_{V_\phi}) (\alpha) \wedge i_{V_\phi}(\beta)  +  \phi^*(\alpha) \wedge i^2_{V_\phi}(\beta)
    \end{gathered}
    \end{equation*} 
\end{definition}

Let us calculate the transformation of the action under the background gauge transformation with $Y$  $\pi_{T[1]X}$-vertical. 
\begin{equation}
    \delta S = \int \sigma^* L_{[Q,Y]}\chi (\dx) + \sigma^*L_{[Q,Y]}\cL
\end{equation}
Let us define the following vector field along $\sigma$
\begin{equation}
    E_\sigma = \dx\sigma^* - \sigma^*Q
\end{equation}
To simplify the formulas we will omit $\sigma^*$ in what follows. We have 
\begin{equation}
\begin{gathered}
    L_{[Q,Y]}\cL = L_Y Q\cL + L_QL_Y\cL = L_Y Q\cL + i_Q d i_Y d\cL = L_Y Q\cL - i_Q L_Y  d\cL = \\ = L_Y Q\cL + i_{E_\sigma} L_Y  d\cL - L_Y  d\cL(\dx)
\end{gathered}
\end{equation}
Since $L_Y  d\cL$ is $d$ exact the last term is a full derivative.

\begin{equation}
    \cL_{[Q,Y]} \chi (\dx) = i_{[Q,Y]} \omega (\dx) + di_{[Q,Y]} \chi (\dx)
\end{equation}
The second term is a total derivative. The first term is rewritten as 
\begin{equation}
    i_{[Q,Y]} \omega(\dx) = L_Y i_Q\omega (\dx) - i_Q L_Y \omega (\dx)
\end{equation}
We have 
\begin{equation}
    L_Y i_Q\omega (\dx) = i_{E_\sigma} L_Y i_Q\omega + i_Q L_Y i_Q\omega = i_{E_\sigma} L_Y i_Q\omega - i_{[Q,Y]}i_Q\omega + L_Yi_Q i_Q \omega
\end{equation}
and
\begin{equation}
\begin{gathered}
    - i_Q L_Y \omega (\dx) = \frac{1}{4}i^2_{E_\sigma} L_Y \omega - \frac{1}{2}L_Y\omega (\dx,\dx) - \frac{1}{2}i_Q i_Q L_Y\omega = \\ = \frac{1}{4}i^2_{E_\sigma} L_Y \omega  + \frac{1}{4}i_Q i_{[Q,Y]} \omega - \frac{1}{2}i_Q L_Yi_Q \omega = \\ = \frac{1}{4}i^2_{E_\sigma} L_Y \omega +  i_Q i_{[Q,Y]} \omega - \frac{1}{2}L_Yi_Q i_Q \omega 
\end{gathered}
\end{equation}
where we have omitted $\frac{1}{2}L_Y\omega (\dx,\dx)$ since it's a full derivative. So all in all:
\begin{equation}
    i_{[Q,Y]} \omega (\dx) = i_{E_\sigma} L_Y i_Q\omega + \frac{1}{2}L_Yi_Q i_Q \omega + \frac{1}{4}i^2_{E_\sigma} L_Y \omega
\end{equation}

The overall variation of the action is 
\begin{equation}
    \delta S = \int L_Y(\frac{1}{2}i_Qi_Q\omega + Q\cL) + i_{E_\sigma}L_Y(i_Q\omega + d\cL) + \frac{1}{4}i^2_{E_\sigma}L_Y\omega
\end{equation}

Suppose that $Y$ is projectable, then $\delta \sigma^* = \sigma^*[Q,Y] - [\dx,y]\sigma^*$. The first term gives all the same while the second gives:
\begin{equation} \label{projY}
    \delta S = -\int L_{[\dx,y]}\sigma^*\chi(\dx) + L_{[\dx,y]}\sigma^*(\cL)
\end{equation}
The first term is 
\begin{equation}
\begin{gathered}
    i_{\dx} L_{[\dx,y]}\sigma^*\chi = - i_{[\dx,[\dx,y]]}\sigma^*\chi  - L_{[\dx,y]}i_{\dx}\sigma^*\chi = \\ = - i_{[\dx,[\dx,y]]}\sigma^*\chi - L_{\dx}L_y i_{\dx}\sigma^*\chi -L_y L_{\dx}i_{\dx}\sigma^*\chi
\end{gathered}
\end{equation}
The first term is trivially zero. For the last term we see that since $\chi$ is of ghost degree $n-1$ then $L_{\dx}i_{\dx}\sigma^*\chi$ is of degree $n+1$ and is therefore zero and we conclude that this part is always a total derivative.

For the second term in \eqref{projY} we have 
\begin{equation}
    L_{[\dx,y]}\sigma^*(\cL) = L_{\dx}L_y\sigma^*\cL + L_y L_{\dx}\sigma^*\cL
\end{equation}
The second term is zero by the same degree resonings as we just had. Therefore making $Y$ projectable does not bring any new constraints on $Y$ to generate the symmetry of the action, the only condition we need is $L_Y\omega \in \cI$.


\section{Independence on the choice of connection}\label{App:Linearization}
Let the presymplectic structure of the system be defined by the presymplectic potential $\chi$. In local coordinates one writes 
\begin{equation}
\begin{gathered}
    \chi = d(\psi^A)\chi_A \\
    \omega = d\chi = \half d(\psi^A)d(\psi^B)((-1)^{|A|+1}\partial_B \chi_A + (-1)^{(|B|+1)|A|}\partial_A \chi_B)
\end{gathered}
\end{equation}
Taking any connection $\Gamma$ the presymplectic potential $\chi_{\Gamma}$ constructed using $V_\Gamma = \phi^A \ddl{}{\psi^A} - \phi^N \phi^M \Gamma^A{}_{NM}\ddl{}{\phi^A}$ is 
\begin{equation}
\begin{gathered}
    \chi_\Gamma = L_V L_V \Pi^* \chi = 2 d(\phi^A) \phi^B \partial_B\chi_A  + d(\psi^A) \phi^B \phi^C \partial_C \partial_B \chi_A  - \\ - d(\psi^A) \phi^N \phi^M \Gamma^B{}_{NM} \partial_B \chi_A - d(\phi^N \phi^M \Gamma^A{}_{NM}) \chi_A 
\end{gathered}
\end{equation}
Using the fact that $\chi$ is defined modulo $d$-exact terms we rewrite the last two terms as:
\begin{equation}
    \begin{gathered}
         -(-1)^{(|A|+1)|B|}\phi^N \phi^M \Gamma^B{}_{NM} d(\psi^A) \partial_B \chi_A -  (-1)^{|A+1|}\phi^N \phi^M \Gamma^A{}_{NM} d(\psi^B)\partial_B\chi_A = \\ = -2\phi^N \phi^M \Gamma^A{}_{NM} d(\psi^B) \omega_{AB}
    \end{gathered}
\end{equation}
We can see that $\chi$ is linear on $\Gamma$, moreover choosing another connection $\Tilde{\Gamma}$, $\Tilde{\Gamma}^A{}_{BC} - \Gamma^A{}_{BC} = R^A{ }_{BC}$ we have 
\begin{equation}
    \chi_{\Tilde{\Gamma}} - \chi_{\Gamma} = -2\phi^N \phi^M R^A{}_{NM} d(\psi^B) \omega_{AB} \in \cI_{\cK}
\end{equation}

It is clear that this term is in $\cI_{\cK}$ because it is proportional to $d\psi^A$ and is annihilated by the kernel of the presymplectic structure.

Therefore, we have proven that the equivalence class of $\chi_{\Gamma}$ is independent of the choice of the connection.

\section{Properties of the closure term}\label{App:closure}
Here we give the details of the calculations, which prove the following equalities:
\begin{enumerate}
    \item $V_i J_j - f^k{ }_{ij}J_k = \pi^*(f_{ij}(x,\theta))$
    \item $Q(V_iJ_j - f^k{ }_{ij}J_k) \sim  QJ + [Q,V_i]J \sim QJ + KJ = 0$\\ and therefore $\dx\pi^*(f_{ij}(x,\theta)) = 0$
    \item $\Tilde{Q}c^ic^j(V_iJ_j - f^k{ }_{ij}J_k) = 0$
\end{enumerate}

1. We have
\begin{equation}
    L_{V_j}i_{V_i} \omega  + L_{V_j}dJ_i + L_{V_j}\beta_i = 0\,.
\end{equation}
Let us rewrite is as 
\begin{equation}
\begin{gathered}
    i_{[V_j,V_i]}\omega + i_{V_i} L_{V_j} \omega + d(L_{V_j}J_i) + L_{V_j}\beta_i = 0 =\\ =-f^k{ }_{ji}(d(J_k) + \beta_k) - i_{V_i}d\beta^j + d(L_{V_j}J_i) + L_{V_j}\beta_i\,,
\end{gathered}
\end{equation}
where we have used $L_{V_j}\omega = -d\beta_j$. It follows 
\begin{equation} \label{dVJ-f}
    d(V_j J_i - f^k{ }_{ji}J_k) =L_{V_i}\beta_j - L_{V_j}\beta_i  + f^k{ }_{ji}\beta_k
\end{equation}
so that  $d(f^k{ }_{ij}J_k - V_i J_j) \in \cI[dx, d\theta]$ which in its turn implies the desired relation.

2. We need to show that distribution $K$ annihilates $J$ which is easily done by applying $i_K$ on its defining equation:
\begin{equation}
    i_K i_{V_i}\omega + KJ_i + i_K\cI = 0
\end{equation}
First term is zero because $V_i$ is vertical, last term is zero because $K$ is vertical.

We also have to show that $QJ_i$ vanishes. Although it is not true for just every $J_i$ we can choose the right representative in the equivalence class. First we note that
\begin{equation}
\begin{gathered}
    L_Qi_{V_i}\omega + L_Q d(J_i) + L_Q\beta_i  = 0\,, \\
    i_{[Q,V_i]} \omega - i_{V_i} L_Q \omega - d(QJ_i) + L_Q \beta_i = 0\,.
\end{gathered}
\end{equation}
Since $L_Q$ preserves the ideal $L_Q \beta_i \in \cI$, since $[Q,V_i] \in K$, $i_{[Q,V_i]}\omega \in \cI$, finally $i_{V_i}L_Q\omega \in \cI$ since $L_Q \omega \in \cI$ and $V_i$ vertical. Therefore we get $d(QJ_i) \in \cI$ from which it follows $QJ_i = \pi^*(\alpha)$. $J_i$ is defined by $\iota_{V_i}\omega+d J_i\in\cI$  modulo functions of the form $\pi^*f$, $f\in \cC^\infty(T[1]X)$ and therefore $J_i$ can be adjusted in such a way that $QJ_i = 0$. ($\alpha$ is of degree $n$ therefore it can be represented as a $\dx$ exact function at least locally.)

3. We are left to show that
\begin{equation} \label{dceVJ}
    ( - \frac{1}{2}f^l{ }_{nm} c^n c^m \ddl{}{c^l})c^ic^j(V_iJ_j - f^k{ }_{ij}J_k) = 0\,.
\end{equation}
Let us examine the equation \eqref{dVJ-f} for this. Applying $i_{V_k}$ to it we get (due to $i_V \beta_i= 0$)
\begin{equation}
    V_k V_j J_i - f^n{ }_{ji}V_kJ_n = 0\,.
\end{equation}
Using $V_iJ_j = -V_jJ_i$ which follows from $i_{V_i}i_{V_j}\omega = - i_{V_j}i_{V_i}\omega$ we rewrite it as 
\begin{equation}
\begin{gathered}
    V_k V_j J_i + f^n{ }_{ji}V_n J_k =  V_k V_j J_i + V_j V_i J_k - V_i V_j J_k = \\ = V_k V_j J_i + V_j V_i J_k + V_i V_k J_j = 0\,.
\end{gathered}
\end{equation}
The expression \eqref{dceVJ} is proportional to $c^n c^m c^k V_n V_m J_k $ therefore it is zero.

\setlength{\itemsep}{0em}
\small
\providecommand{\href}[2]{#2}\begingroup\raggedright\endgroup

\end{document}